\documentclass[journal,10pt,twoside,twocolumn]{IEEEtran}

\usepackage{algorithm2e}
\usepackage{amsmath,amsfonts} 
\usepackage{amssymb, amsthm}
\usepackage{graphicx}
\usepackage{subfigure}
\usepackage{booktabs}
\usepackage{multirow}
\usepackage{mathrsfs}
\usepackage{cite}
\usepackage{units}
\usepackage{cancel}
\usepackage{upgreek}
\usepackage{color}
\usepackage{xcolor}
\DeclareGraphicsExtensions{.pdf,.jpeg,.png,.epbibdesks}
\usepackage{epstopdf}

\renewcommand{\eqref}[1]{(\ref{#1})}
\newcommand{\figref}[1]{\mbox{\figurename~\ref{#1}}}

\renewcommand{\mid}{\ensuremath{\,|\,}}

\newcommand{\abs}[1]{\ensuremath{\left|#1\right|}}

\newcommand{\setO}{\ensuremath{\mathcal{O}}}

\newcommand{\bma}{\ensuremath{\mathbf{a}}}

\newcommand{\bme}{\ensuremath{\mathbf{e}}}
\newcommand{\bmf}{\ensuremath{\mathbf{f}}}
\newcommand{\bmg}{\ensuremath{\mathbf{g}}}
\newcommand{\bmh}{\ensuremath{\mathbf{h}}}

\newcommand{\bmn}{\ensuremath{\mathbf{n}}}

\newcommand{\bmr}{\ensuremath{\mathbf{r}}}
\newcommand{\bms}{\ensuremath{\mathbf{s}}}

\newcommand{\bmu}{\ensuremath{\mathbf{u}}}
\newcommand{\bmv}{\ensuremath{\mathbf{v}}}
\newcommand{\bmw}{\ensuremath{\mathbf{w}}}
\newcommand{\bmx}{\ensuremath{\mathbf{x}}}
\newcommand{\bmy}{\ensuremath{\mathbf{y}}}

\newcommand{\bA}{\ensuremath{\mathbf{A}}}

\newcommand{\bD}{\ensuremath{\mathbf{D}}}
\newcommand{\bE}{\ensuremath{\mathbf{E}}}
\newcommand{\bF}{\ensuremath{\mathbf{F}}}
\newcommand{\bG}{\ensuremath{\mathbf{G}}}
\newcommand{\bH}{\ensuremath{\mathbf{H}}}
\newcommand{\bI}{\ensuremath{\mathbf{I}}}

\newcommand{\bL}{\ensuremath{\mathbf{L}}}

\newcommand{\bW}{\ensuremath{\mathbf{W}}}
\newcommand{\bX}{\ensuremath{\mathbf{X}}}

\newcommand{\bZero}{\ensuremath{\mathbf{0}}}

\newcommand{\bDelta}{\ensuremath{\mathbf{\Delta}}}

\newcommand{\MT}{{\ensuremath{U}}}
\newcommand{\MR}{{\ensuremath{B}}}

\newcommand{\Es}{{\ensuremath{E_{\mathrm{s}}}}}

\newtheorem{thm}{Theorem}
\newtheorem{lem}[thm]{Lemma}
\makeatletter
\renewenvironment{proof}[1][\proofname]{\par
  \pushQED{\qed}%
  \normalfont \topsep6\p@\@plus6\p@\relax
  \trivlist
  \item[\hskip\labelsep
        \scshape
    #1\@addpunct{.}]\ignorespaces
}{%
  \popQED\endtrivlist\@endpefalse
}
\makeatother

\newcommand{\revision}[1]{#1}

\title{Large-Scale MIMO Detection for 3GPP LTE: Algorithms and FPGA Implementations}

\author{Michael Wu, Bei Yin, Guohui Wang, Chris Dick, Joseph R. Cavallaro, and  Christoph~Studer

\thanks{M.~Wu and B.~Yin  contributed equally to the paper. }
\thanks{Parts of this paper for a large-scale, point-to-point MIMO-OFDM system have been presented at the IEEE International Symposium on Circuit and Systems (ISCAS) \cite{Wu2012} and the IEEE International Conference on Acoustics, Speech, and Signal Processing (ICASSP)~\cite{Yin2013}.}
\thanks{\revision{M.~Wu, B.~Yin, G.~Wang, and J.~Cavallaro are with the Dept.~of ECE, Rice University, Houston, TX (e-mail: \{mbw2,\,by2,\,wgh,\,cavallar\}@rice.edu).}} 
\thanks{\revision{C.~Dick is with Xilinx, Inc., San Jose, CA (e-mail: chris.dick@xilinx.com).}}
\thanks{\revision{C.~Studer is with the School of ECE, Cornell University, Ithaca, NY (e-mail: studer@cornell.edu)}}
\thanks{This work was supported in part by Xilinx and by the US National Science Foundation under grants CNS-1265332, ECCS-1232274, EECS-0925942, and CNS-0923479.}
 
}

\begin{document}

\maketitle

\begin{abstract}
\revision{Large-scale (or massive) multiple-input multiple-output (MIMO) is expected to be one of the key technologies in next-generation multi-user cellular systems based on the upcoming 3GPP LTE Release~12 standard, for example. 
In this work, we propose---to the best of our knowledge---the first VLSI design enabling high-throughput data detection in single-carrier frequency-division multiple access (SC-FDMA)-based large-scale MIMO systems.}
\revision{We propose a new approximate matrix inversion algorithm relying on a Neumann series  expansion, which substantially reduces the complexity of linear data detection. We analyze the associated error,  and we compare its performance and complexity to those of an exact linear detector. We present corresponding VLSI architectures, which perform exact and approximate soft-output detection for large-scale MIMO systems with various antenna/user configurations. 
Reference implementation results for a Xilinx \mbox{Virtex-7} XC7VX980T FPGA show that our designs are able to achieve more than $\mathbf{600}$\,Mb/s for a $\mathbf{128}$ antenna, $\mathbf{8}$ user 3GPP LTE-based large-scale MIMO system.} We finally provide a performance/complexity trade-off comparison using the presented FPGA designs, which reveals that the \revision{detector} circuit of choice is determined by the ratio between BS antennas and users, as well as the desired error-rate performance. 

\end{abstract}



\section{Introduction}

Multiple-input multiple-output (MIMO) in combination with spatial multiplexing~\cite{Paulraj2008} builds the foundation of most modern wireless communication standards, such as 3GPP LTE~\cite{3GPPLTE,SesiaLTE,3GPPLTEA} or IEEE 802.11n\cite{IEEE802.11n}. 
\revision{MIMO technology offers 
significantly higher data rates over single-antenna systems by transmitting multiple data streams concurrently and in the same frequency band.} \revision{Conventional MIMO wireless systems, however, already start to approach their throughput limits.} Consequently, the deployment of novel transceiver technologies is of paramount importance in order to meet the ever-growing demand for higher data rates, better link reliability, and improved coverage, without further increasing the communication bandwidth \cite{Marzetta2010,Rusek2012,Nam2013}.

\subsection{Blessing and Curse of Massive MIMO}

Large-scale (or massive) MIMO  is an emerging technology, which postulates the use of antenna arrays having orders of magnitude more elements at the base station (BS) compared to conventional (small-scale) MIMO systems, while  serving tens of users simultaneously and in the same frequency band~\cite{Marzetta2010}. This technology promises significant improvements in terms of spectral efficiency, link reliability, and \revision{coverage compared to conventional (small-scale) systems} \cite{Rusek2012, Huh2011,Ngo2012}. 

Unfortunately, the promised benefits of large-scale MIMO come at the cost of significantly increased computational complexity in the BS, as opposed to small-scale MIMO systems, which commonly deploy $2$-to-$4$ antennas at both ends of the wireless link.
In particular, data detection in the large-scale MIMO uplink is expected to be \revision{among} the most critical tasks in terms of complexity and power consumption, as the presence of hundreds of antennas at the BS and a large number of users will increase the computational complexity by orders of magnitude.
In addition, current cellular systems, such as 3GPP-LTE~\cite{3GPPLTE,SesiaLTE} or LTE-Advanced (LTE-A)~\cite{3GPPLTEA}, rely on single-carrier frequency division multiple access (SC-FDMA), which further increases the dimensionality \revision{(and hence the complexity)} of the underlying detection problem. 
\revision{As a consequence,} optimal \revision{data} detection methods, such as maximum-likelihood (ML) detection~\cite{agrell2002closest,ABurgThesis,burg2005vlsi} or soft-output sphere decoding~(SD)~\cite{wong2002vlsi,HB03,jsac07}, whose \revision{(average)} computational complexity scales exponentially in the number of transmitted data streams~\cite{jalden2005complexity,seethaler2011complexity}, would simply result in prohibitive complexity. 
\revision{Hence,} one has to resort to low-complexity (but sub-optimal) linear detection schemes~\cite{Rusek2012} or stochastic detection algorithms~\cite{Datta2012} that deliver acceptable error-rate performance and scale favorably to the high-dimensional detection problems faced in SC-FDMA-based large-scale MIMO systems.

\subsection{Contributions}
This paper addresses the complexity issue of data detection in SC-FDMA-based large-scale MIMO systems in the uplink, i.e., where multiple users communicate with the BS.
\revision{We focus on linear soft-output detection in combination with a new approximate matrix inversion method relying on a  Neumann series expansion, which significantly reduces the computational complexity compared to that of an exact matrix inversion method.}
We analyze the implementation trade-offs associated with approximate and exact linear data detection in the large-scale MIMO uplink, and we show analytically that the approximation error caused by the proposed approximate inversion method depends on the ratio between BS antennas and users. 
\revision{We show that the proposed approximation performs well for medium to large ratios between BS antennas and users, while exact linear detection is advantageous for small antenna ratios.}
We  present reference FPGA designs for both, the \revision{approximate and exact matrix inversion,} and for various antenna configurations, which enables us to characterize the associated hardware complexity vs.~error-rate performance trade-offs. The resulting FPGA designs are---to the best of our knowledge---the first data detection engines for massive MIMO systems reported in the \revision{open} literature that achieve a peak uplink throughput exceeding the $300$\,Mb/s specified in \revision{3GPP LTE-Advanced} operating at 20\,MHz bandwidth~\cite{3GPPLTEA}.

\subsection{Notation}
\label{sec:notation}
Lowercase boldface letters stand for column vectors; uppercase boldface letters designate matrices. For a matrix~\bA, we denote its transpose and conjugate transpose  $\bA^T$ and  $\bA^H$, respectively. 
The entry in the $k^\text{th}$ row and $\ell^\text{th}$ column of a matrix $\bA$ \revision{is} denoted by $A_{k,\ell}$; the $k^\text{th}$ entry of a vector $\bma$ is designated by $a_k$. The Frobenius norm and $\ell_2$-norm of a matrix~$\bA$ and vector~$\bma$ \revision{are} denoted by $\|\bA\|_F$ and $\|\bma\|_2$, respectively. 
The $M\times M$ identity matrix is denoted by $\bI_M$, and~$\bF_M$ refers to the  $M\times M$ discrete Fourier transform~(DFT) matrix, normalized as $\bF_M^H\bF_M=\bI_M$.
In order to simplify notation, we make frequent use \revision{of} the superscript $(\cdot)^{(i,j)}$ to indicate the $i^\text{th}$ base-station antenna and $j^\text{th}$ user; the subscript~$(\cdot)_w$ \revision{designates} the SC-FDMA subcarrier index.

\subsection{Paper Outline}

The remainder of the paper is organized as follows. 
Section~\ref{sec:system}  introduces the uplink system model and outlines \revision{the basics of} linear detection for \revision{SC-FDMA-based} systems. 
The approximate matrix inversion approach, a corresponding error analysis, and an error-rate performance/complexity comparison \revision{are} shown in Section~\ref{sec:algorithm}.
Section~\ref{sec:architecture} details \revision{our} VLSI architecture. Section~\ref{sec:implementation} provides \revision{reference} FPGA implementation results and a trade-off analysis. 
We conclude in Section~\ref{sec:conclusions}. All proofs are relegated to the Appendices.


\section{Large-Scale MIMO in LTE Uplink}
\label{sec:system}

We next introduce the LTE uplink model and \revision{present a new and} efficient method for linear soft-output minimum mean-square error (MMSE) detection in SC-FDMA-based systems. 

\subsection{LTE Uplink Model}

We consider the large-scale multi-user (MU) MIMO uplink with~$\MR$ antennas at the base-station (BS) communicating with $\MT\leq\MR$ single-antenna users.\footnote{More generally, instead of having $\MT$ single-antenna users, the proposed system may equivalently support $\MT$ spatial streams, which can, for example, be shared among a smaller number of user terminals that are equipped with more than one antenna.}
 To reduce \revision{the} peak-to-average power ratio of \revision{the} user equipment, LTE uplink employs SC-FDMA (short for single-carrier frequency division multiple access)~\cite{SesiaLTE}. The $\MT$ users first encode their own transmit \revision{bits} using channel \revision{encoders} and then, map the coded bit \revision{stream} to \revision{time-domain} constellation points in the finite alphabet $\setO$ with cardinality $M=\abs{\setO}$ and average transmit power~$E_{s}$ per symbol. 
\revision{An $L$-point discrete Fourier transform (DFT) block\footnote{In practice, the DFT and inverse DFT are carried out by fast (inverse) Fourier transform (I/FFT) units.} is used to perform modulation of these time-domain symbols onto orthogonal frequency bands. The $L$ time-domain constellation points for the $i^\text{th}$ user are subsumed in the vector $\bmx^{(i)}=\big[x^{(i)}_1,\ldots,x^{(i)}_L\big]^T$. The output of the DFT block, namely the frequency-domain symbol, is defined as $\bms^{(i)}  = \big[\bms^{(i)}_1,\ldots,\bms^{(i)}_L\big]^T = \bF_L\bmx^{(i)}$. } Subsequent processing performed for each user corresponds \revision{to that} of conventional orthogonal frequency-division multiplexing (OFDM) transmission~\cite{OFDM2004}. Specifically, for each user, the frequency-domain symbols are first mapped onto data-carrying subcarriers and then, transformed back to \revision{the} time domain with an inverse  DFT (IDFT). After \revision{prepending} the cyclic prefix to the time-domain symbols, all $\MT$ users  transmit their time-domain signals simultaneously over the wireless channel.

At the BS, each receive antenna obtains a mixture of the time-domain signals from all users. For data detection, the time-domain signals received at each antenna are first transformed back into the frequency domain using a DFT. The data-carrying symbols are then extracted from the DFT's output. Assuming a sufficiently long cyclic-prefix~(i.e., longer than the delay spread of the channel's impulse response), the received \revision{frequency-domain} symbols can be modeled using the standard input-output relation $\bmy = \bH\bms +\bmn$,
with the following definitions:
\begin{align*}
\bmy &= \left[\begin{array}{c}
\bmy^{(1)}\\[-0.2cm]
\vdots\\
\bmy^{(\MR)}
\end{array}\right], \quad  \bH=\left[\begin{array}{ccc}
\bH^{(1,1)} & \cdots & \bH^{(1,\MT)}\\[-0.2cm]
\vdots & \ddots & \vdots\\
\bH^{(\MR,1)} & \cdots & \bH^{(\MR,\MT)}
\end{array}\right], \\ 
\bms&=\left[\begin{array}{c}
\bms^{(1)}\\[-0.2cm]
\vdots\\
\bms^{(\MT)}
\end{array}\right], \quad \text{and}\,\,\, \bmn = \left[\begin{array}{c}
\bmn^{(1)}\\[-0.2cm]
\vdots\\
\bmn^{(\MR)}
\end{array}\right]. 
\end{align*}
Here, the vector $\bmy^{(i)} = \big[y^{(i)}_1, \ldots, y^{(i)}_L\big]^T$ contains the received symbols on the $i^\text{th}$ antenna in the frequency domain, where $y^{(i)}_w$ is the symbol received on the \revision{$w^\text{th}$} subcarrier of the $i^\text{th}$ antenna.  The $L\times L$ diagonal matrix $\bH^{(i,j)}=\text{diag}\big(h^{(i,j)}_1,\ldots,h^{(i,j)}_L\big)$ contains the channel's frequency response of length $L$ between the $i^\text{th}$ receive antenna and $j^\text{th}$ transmit antenna on its main diagonal, and $\bmn^{(i)} = \big[n^{(i)}_1, \ldots, n^{(i)}_L\big]^T$ models thermal noise at the $i^\text{th}$ receive antenna in the frequency domain. The entries of  the vector~$\bmn^{(i)}$ are assumed to be i.i.d.\ zero-mean Gaussian with variance $N_0$ per complex entry. 

\subsection{Linear MMSE Detection}
\label{sec:lindetection}

The task of a data detector for MIMO systems is to compute soft-estimates in the form of log-likelihood ratio (LLR) \revision{values} for each coded bit, given the channel matrix\footnote{In practice, channel-state information is acquired using pilot sequences specified by the standard~\cite{3GPP_TS_36.211_v8.6.0}. For the sake of simplicity, we assume perfect channel state information (CSI) throughout the paper. An investigation of the  impact of imperfect CSI on the error-rate performance is left for future work.} $\bH$ and receive vector~$\bmy$. 
In order to arrive at low computational complexity for data detection in SC-FDMA-based large-scale MIMO systems, we focus exclusively on linear soft-output detection \cite{khan2009lte}. 
Linear detection for SC-FDMA mainly consists of the following two steps: (i) \emph{channel equalization} to generate estimates of the frequency domain symbols,  and (ii) \emph{soft-output computation} to generate LLRs from the equalized frequency domain symbols. 
Both of these steps are detailed next. 

\subsubsection*{(i) Channel equalization}
The most common approach to linear MIMO detection is the minimum-mean square error (MMSE) equalizer, which computes equalized frequency-domain symbols as   $\hat{\bms}=\bW\bmy$ with the MMSE equalization matrix  defined as follows~\cite{Paulraj2008}:
\begin{align*}
\bW=\!\left(\bH^H\bH+{N_0}{\Es^{\!-1}}\bI_{L\MT}\right)^{-1}\bH^H.
\end{align*}
Since the effective channel matrix $\bH$ is built from diagonal $L\times L$ submatrices, we can apply MMSE equalization on a per-subcarrier basis. 
Specifically, the received frequency symbols on the $w^\text{th}$  subcarrier in the frequency domain can be modeled as $\bmy_w = \bH_w\bms_w +\bmn_w$, where 
\begin{align*}
\bmy_w &= \left[\begin{array}{c}
y_w^{(1)}\\[-0.2cm]
\vdots\\
y_w^{(\MR)}
\end{array}\right], \quad \bH_w = \left[\begin{array}{ccc}
h^{(1,1)}_w & \cdots & h^{(1,\MT)}_w\\[-0.2cm]
\vdots & \ddots & \vdots\\
h^{(\MR,1)}_w & \cdots & h^{(\MR,\MT)}_w
\end{array}\right], \\
\bms_w&=\left[\begin{array}{c}
\!\!s^{(1)}_w,\ldots,
s^{(\MT)}_w\!\!
\end{array}\right]^T, \quad \text{and}\quad\bmn_w=\left[\begin{array}{c}
\!\!n^{(1)}_w,\ldots,
n^{(\MR)}_w\!\!
\end{array}\right]^T\!\!.
\end{align*}
Here, $y^{(i)}_w$ is the frequency symbol received on the $w^\text{th}$ subcarrier for the $i^\text{th}$ antenna, \revision{and} $h^{(i,j)}_w$ is the frequency gain (or attenuation) on the $w^\text{th}$ subcarrier between the $i^\text{th}$ receive antenna and $j^\text{th}$ transmit antenna. The scalar $s^{(j)}_w$ denotes the symbol transmitted by the $j^\text{th}$ user on the $w^\text{th}$ subcarrier; the scalar $n^{(i)}_w$ models thermal noise at the $i^\text{th}$ receive antenna on the~$w^\text{th}$ subcarrier.
With this reformulation, the equalized symbols on the $w^\text{th}$ subcarrier are given by $\hat{\bms}_w= \bW_w\bmy_w$, with the per-subcarrier MMSE equalization matrix defined as
\begin{align} \label{eq:subcarrierMMSE}
\bW_w= \!\left(\bH_w^H\bH_w+{N_0}{\Es^{\!-1}}\bI_{\MT}\right)^{-1}\bH_w^H.
\end{align}

A key method to arrive at low-complexity linear MMSE detection was put forward in~\cite{Studer2011}. This approach  first computes the matched-filter~(MF) output as $\bmy^\text{MF}_w = \bH_w^H\bmy_w$ and the Gram matrix $\bG_w=\bH_w^H\bH_w$ for each subcarrier \revision{$w$}, followed by forming the regularized Gram matrix $\bA_w=\bG_w+{N_0}{\Es^{\!-1}}\bI_{\MT}$. The equalized symbols per subcarrier are \revision{then} computed as $\hat{\bms}_w =\bA_w^{-1}{\bmy}^\text{MF}_w$, which requires the \emph{explicit} computation of \revision{a} $\MT\times\MT$-dimensional matrix inverse.\footnote{We are aware of the fact that the estimate $\hat{\bms}_w$ could be computed without forming the explicit inverse  $\bA_w^{-1}$, e.g., via the Cholesky decomposition combined with  forward/backward substitution~\cite{GV96}. However,  soft-output detection as performed here requires the  explicit inverse $\bA_w^{-1}$ to compute the post-equalization SINR (see Section~\ref{subsubsection:exactllr}).}

\subsubsection*{(ii) LLR computation}\label{subsubsection:exactllr}
To obtain  symbol estimates in the time domain, the MMSE detector performs an IDFT on the equalized frequency domain symbols for each user. The time-domain symbol estimates for the $i^\text{th}$ user are given by $\hat{\bmx}^{(i)}=\bF_L^H \hat{\bms}^{(i)}$, where $\bF_L^H$ is the IDFT matrix and $\hat{\bmx}^{(i)}=\big[\hat{x}^{(i)}_1,\ldots,\hat{x}^{(i)}_L\big]^T$ contains the time-domain symbol estimates of the symbols transmitted by the $i^\text{th}$ user. To extract LLRs from the time-domain symbol estimates, we approximate each estimate as an independent Gaussian random variable. 
In particular,  the estimated $t^\text{th}$ symbol transmitted from the $i^\text{th}$ user is modeled as $\hat{x}^{(i)}_t=\mu^{(i)} x^{(i)}_t+{e^{(i)}_t}$, where $\mu^{(i)}$ is the effective channel gain and~$e^{(i)}_t$ is the post-equalization noise-plus-interference~(NPI) variance.
Let $\nu_i^2$ be the variance of~$e^{(i)}_t$ and $b$ be the bit index of the LLR associated with the  $t^\text{th}$ symbol transmitted from the $i^\text{th}$ user. With this model, the max-log LLRs can  be computed as~\cite{fossorier1998equivalence,Studer2011} 
\begin{align}\label{eq:maxlogllr}
L^{(i)}_{t}(b) = \rho^2_i\!\left(\min_{a\in\setO_b^{0}}\left|{\frac{\hat{x}^{(i)}_t}{\mu^{(i)}}}-a\right|^2 - \min_{a'\in\setO_b^{1}}\left|{\frac{\hat{x}^{(i)}_t}{\mu^{(i)}}}-a'\right|^2\right),
\end{align}
where $\rho^2_i={\left(\mu^{(i)}\right)^2}/{\nu_i^2}$ is the post-equalization signal-to-noise-plus-interference ratio (SINR), and $\setO_b^{0}$ and $\setO_b^{1}$ correspond to the sets of constellation symbols for which the $b^\text{th}$ bit equals to $0$ and $1$, respectively. 

In order to obtain an explicit formulation of the effective channel gain $\mu^{(i)}$ as well as the NPI variance~$\nu^2_i$, we can write the $t^\text{th}$ symbol estimate of the $i^\text{th}$ user as follows:
\begin{align*}
\hat{x}^{(i)}_t &= \bmf^{H}_t\hat{\bms}^{(i)} = \bmf^{H}_t\bW^{(i,:)}\bmy.
\end{align*}
Here, $\bW^{(i,:)}=\left[\bW^{(i,1)},\ldots,\bW^{(i,\MR)}\right]$ is a horizontal concatenation of the $i^\text{th}$ block row of (diagonal) submatrices of $\bW$. The row vector $\bmf^{H}_t$ corresponds to the $t^\text{th}$ row of the IDFT matrix $\bF_L^H$. 
Let $\bH^{(:,j)}=\!\left[\bH^{(1,j)},\ldots,\bH^{(\MR,j)}\right]^T$ be the horizontal concatenation of the $j^\text{th}$ block column of (diagonal) submatrices of $\bH$, consisting of the frequency-domain channel responses between the receive antennas and the transmit antenna associated with the $j^\text{th}$ user. 
We first compute the effective channel gain: 
\begin{align*}
\mu^{(i)} x^{(i)}_t &=\mathsf{E}\!\left[\bmf^{H}_t\bW^{(i,:)}\bmy\mid x^{(i)}_t \right]  
 = L^{-1}\mathrm{tr}(\bW^{(i,:)}\bH^{(:,i)})x^{(i)}_t.
\end{align*}
Since $\bW^{(i,j)}$ and $\bH^{(i,j)}$ are both diagonal matrices, we can write $\mu^{(i)}$ as a sum of per-subcarrier operations. In particular, let $\bmw^H_{i,w}$  be the $i^\text{th}$ row of $\bW_w$ and $\bmh_{i,w}$ be the $i^\text{th}$ column of $\bH_w$. Then, we obtain the effective channel gain as
\begin{align}\label{eq:mu}
 \mu^{(i)} = L^{-1}\sum_{w=1}^L{\bmw^H_{i,w}\bmh_{i,w}}.
\end{align}

We next compute the post-equalization NPI variance $\nu_i^2$ of the residual noise plus interference as
\begin{align*}
&\nu_i^2 = \mathsf{E}\!\left[\abs{\hat{x}^{(i)}_{t}}^2\right]- \mathsf{E}\!\left[\abs{\mu^{(i)} x^{(i)}_t}^2\right] \\
  &= \mathsf{E}\!\left[\bmf^{H}_t\bW^{(i,:)}(\bH\bms +\bmn)(\bH\bms +\bmn)^H(\bW^{(i,:)})^H\bmf_t\right]\!-\! E_s\!\abs{\mu^{(i)}}^2\!\!.
 \end{align*}
The MMSE equalization matrix can be written in two ways~\cite{Studer2011}, i.e., either 
\begin{align*}
\bW&=\!\left(\bH^H\bH+{N_0}{\Es^{\!-1}}\bI_{L\MT\times L\MT}\right)^{-1}\bH^H \,\,\, \text{or} \\
\bW&=\!\bH^H\left(\bH\bH^H+{N_0}{\Es^{\!-1}}\bI_{L\MR\times L\MR}\right)^{-1}.
\end{align*}
\revision{Hence, we have $\bW\!\left(E_s\bH\bH^H + N_0\bI_{L\MR\times L\MR}\right)=E_s\bH^H$; this allows us to rewrite the post-equalization NPI in compact form as follows~\cite{Studer2011}:
\begin{align}\label{eq:approxexactllr} \nu_i^2=E_s\mu^{(i)}-E_s\!\abs{\mu^{(i)}}^2. \end{align}
We emphasize that both parameters $\mu^{(i)}$ and $\nu^2_i$ are functions of $\bA^{-1}_w$, $\forall w$. Consequently,  \revision{an explicit computation} of the inverses $\bA_w^{-1}$, $\forall w$, is necessary for the computation of LLR values using the approach detailed above.}

\section{Approximate MMSE Detection \\ via Neumann Series Expansion}
\label{sec:algorithm}

The computation of all per-subcarrier inverses $\bA^{-1}_{w}$, $\forall w$, in \eqref{eq:subcarrierMMSE} is responsible for the main computational complexity of linear MMSE detection in SC-FDMA-based large-scale MIMO systems. 
For a conventional small-scale LTE uplink scenario, i.e., where the number of receive antennas $\MR$ and users $\MT$ is small (on the order of $\MT,\MR\leq 6$), existing VLSI designs  for linear detection, such as~\cite{burg2006algorithm,rao2010low,luethi2008gram}, compute the exact inverse explicitly. 
For  large-scale MIMO systems with a large number of users $\MT$ however, the computation of the inverse $\bA_w^{-1}$ can quickly result in excessive complexity. 
Hence, practical solutions for large-scale MIMO detection in LTE necessitate \revision{low-complexity} matrix inversion methods---a corresponding approximate solution is proposed next.

\subsection{Neumann Series Approximation}

For large-scale MIMO systems, where the number of receive antennas is larger than the number of \revision{single-antenna} users, i.e., for $\MT\ll\MR$, the Gram matrices $\bG_w$, and, consequently~$\bA_w$, become diagonally dominant~\cite{Rusek2012}. In fact, for \revision{i.i.d.\ Gaussian} channel matrices $\bH_w$ (with properly normalized entries) and in the large antenna limit, \cite{Marzetta2010} shows that $\bG_w\rightarrow\bI_\MT$. 
Inspired by this central property of large-scale MIMO, one can derive a low-complexity approximation of the inverse. In particular, let $\bA_w\approx\bD_w$, where $\bD_w$ is the main diagonal  of $\bA_w$. 
As a result, the inverse $\bA^{-1}_w$ can be approximated by $\bD_w^{-1}$, which requires evidently much lower complexity than that of the exact inverse. Unfortunately, for realistic antenna/user configurations, such a crude approximation would cause a significant performance loss.
Hence, to arrive at an accurate approximation of the inverse at low computational complexity, we propose to \revision{use a Neumann series expansion.} 

We start by rewriting the inverse  $\bA_w^{-1}$ with the following Neumann series \revision{expansion}~\cite{Stewart1998}:
\begin{align} \label{eq:origseries}
\bA_w^{-1}=\sum_{n=0}^{\infty}\left(\mathbf{\mathbf{X}}^{-1}\left(\mathbf{X}-\bA_w\right)\right)^{n}\mathbf{\mathbf{X}}^{-1},
\end{align}
which holds if $\lim_{n\to\infty}(\mathbf{I}-\mathbf{X}^{-1}\bA_w)^{n}=\bZero_{\MT\times\MT}$ is satisfied.
By decomposing the regularized Gram matrix~$\bA_w$ such that $\bA_w=\bD_w+\bE_w$, where $\bD_w$ is the main diagonal of $\bA_w$ and~$\bE_w$ \revision{is} the hollow, regularized Gram matrix, we can rewrite  the Neumann series  in \eqref{eq:origseries} as
\begin{align}\label{eq:series}
\bA^{-1}_w =  \sum_{n=0}^{\infty}(-\bD_w^{-1}\bE_w)^{n}\bD_w^{-1},
\end{align}
\revision{where we substitute $\bX$ in \eqref{eq:origseries} by~$\bD_w$}. Note that if $\lim_{n\to\infty}(-\bD^{-1}_w\bE_w)^{n}=\bZero_{\MT\times\MT}$, then the series expansion in~\eqref{eq:series} is guaranteed to converge.

The key idea of the proposed approximate inversion method is to keep only the first $K$ terms of the Neumann series~\eqref{eq:series}. Concretely, we compute a $K$-term approximation as follows:
\begin{align} \label{eq:approximation}
\widetilde{\bA}^{-1}_{w\mid K} =  \sum_{n=0}^{K-1}(-\bD_w^{-1}\bE_w)^{n}\bD_w^{-1},
\end{align}
which can be computed at low computational complexity for approximations consisting of only a few Neumann series terms, i.e., for small values of $K$.
With this approximation, the resulting  \emph{approximate} MMSE equalization matrix is given by $\widetilde{\bW}_{w\mid K} = \widetilde{\bA}^{-1}_{w\mid K}\bH_w^H$.
For $K=1$, we obtain 
$\widetilde{\bA}^{-1}_{w\mid 1}  = \bD_w^{-1}$,
which is simply a scaled version of the MF detector, as $\widetilde{\bW}^{-1}_{w\mid 1}  = \bD_w^{-1}\bH_w^H$. We emphasize that the row-wise scaling induced by $\bD_w^{-1}$ does not affect the detection process, as long as~$\bD_w^{-1}$ exists. Hence, the proposed approximation \eqref{eq:approximation} simply coincides with the MF detector for $K=1$.
For $K=2$, we obtain
$\widetilde{\bA}^{-1}_{w\mid2}  = \bD_w^{-1} - \bD_w^{-1}\bE_w\bD_w^{-1}$,
whose computational complexity only scales with $O(\MT^2)$ operations; this \revision{is} in contrast to the $O(\MT^3)$ complexity scaling required by computing an exact inverse. Hence, a second-order Neumann series approximation can be obtained at lower computational  complexity. For $K=3$, we obtain
\begin{align} \label{eq:threetermapprox}
\widetilde{\bA}^{-1}_{w\mid3}  = \bD_w^{-1} - \bD_w^{-1}\bE_w\bD_w^{-1}+\bD_w^{-1}\bE_w\bD^{-1}\bE_w\bD^{-1}_w,
\end{align}
whose complexity scales with $O(\MT^3)$, which is equivalent to that of an exact inverse. Nevertheless, evaluating~\eqref{eq:threetermapprox} requires fewer arithmetic operations than an explicit evaluation of $\bA^{-1}$. 
Note that for $K\geq4$, computing the exact inverse can be of lower complexity than the proposed approximation, e.g., when using a Cholesky factorization (see Section~\ref{sec:sims}).\footnote{For approximations with $K=2^n$ terms and $n\geq2$, efficient ways of evaluating \eqref{eq:approximation} exist. In particular, a clever re-arrangement and factorization of terms yields solutions which only require $2(n-1)$ matrix multiplications.}

\subsection{Analysis of the Approximation Error}
\label{sec:approxerror}

We next analytically characterize the error induced by the approximate inverse \eqref{eq:approximation} for MMSE estimation. To this end, we define the approximation error as
$\bDelta_{w\mid K}	= \bA_w^{-1} - \widetilde{\bA}_{w\mid K}^{-1}$, 
which is equivalent to 
\begin{align*}
		\bDelta_{w\mid K}	& = \sum_{n=K}^{\infty}(-\bD_w^{-1}\bE_w)^{n}\bD_w^{-1} \\
		& =\left(-\bD_w^{-1}\bE_w\right)^{K}\sum_{n=0}^{\infty}(-\bD_w^{-1}\bE_w)^{n}\bD_w^{-1} \\
		& =\left(-\bD_w^{-1}\bE_w\right)^{K}\bA_w^{-1}.
\end{align*}
Now, consider the situation of using the approximate $\widetilde{\bA}^{-1}_{w\mid K}$ in place of $\bA_w^{-1}$ to compute the equalized frequency-domain symbols, i.e., 
\begin{align*} 
\hat\bms_{w\mid K} = \widetilde{\bA}^{-1}_{w\mid K} \bH_w^H\bmy_w 
= \bA_w^{-1}\bmy^\text{MF}_w-\bDelta_{w\mid K}\bmy^\text{MF}_w
\end{align*}
with $\bmy^\text{MF}_w= \bH_w^H\bmy_w$ and $\hat\bms_w=\bA_w^{-1}\bmy^\text{MF}_w$ being the exact estimate.
We can bound the $\ell_2$-norm of the \emph{residual estimation error} resulting from this approximate equalization \revision{by}
\begin{align} 
\|\bDelta_{w\mid K} \bmy^\text{MF}_w\|_2 &= \|(-\bD^{-1}_w\bE_w)^K\bA^{-1}_w \bmy^\text{MF}_w\|_2\notag \\
& \leq \|(-\bD^{-1}_w\bE_w)^K\|_{F} \|\bA^{-1}_w \bmy^\text{MF}_w\|_2 \notag \\
& \leq \|\bD^{-1}_w\bE_w\|^K_{F} \|\hat\bms_w\|_2. \label{eq:errorbound1}
\end{align}
From \eqref{eq:errorbound1}, we see that  if the condition
\begin{align} \label{eq:convergencecriterion}
\|\bD^{-1}_w\bE_w\|_{F}<1
\end{align}
is satisfied, then the approximation error approaches zero exponentially fast \revision{as} \revision{$K\to\infty$}. 
Moreover, one can show that~\eqref{eq:convergencecriterion} is a sufficient condition for~\eqref{eq:series} to converge.

We now show that the condition $\|\bD^{-1}_w\bE_w\|_{F}<1$ is satisfied \revision{with high probability for large-scale MIMO systems with a larger number of BS antennas $\MR$ than users $\MT$}, and if the entries of $\bH_w\in\mathbb{C}^{\MR\times\MT}$ are assumed to be i.i.d.\ circularly symmetric complex Gaussian with unit variance. 
More specifically, we arrive at a condition that only depends on $\MT$ and $\MR$ for (i) the proposed Neumann series to converge and (ii) the residual approximation error \eqref{eq:errorbound1} to be small.
The following theorem, which is proven in Appendix \ref{app:theoremproof}, makes this behavior explicit. 

\begin{thm}
\label{theorem}
Let $\MR>4$ and the entries of $\bH_w\in\mathbb{C}^{\MR\times\MT}$ be  i.i.d.\ circularly symmetric complex Gaussian with unit variance. Then, we have 
\begin{align} \label{eq:thmcondition}
&\mathrm{Pr}\!\left\{\|\bD_w^{-1}\bE_w\|^K_{F}< \alpha\right\} \notag \\
& 
\,\,\quad \geq 1- \frac{\left(\MT^2\!-\!\MT\right)}{\alpha^{\frac{2}{K}}}\sqrt{\frac{2\MR(\MR\!+\!1)}{(\MR\!-\!1)(\MR\!-\!2)(\MR\!-\!3)(\MR\!-\!4)} }.
\end{align}
\end{thm}

We emphasize that this theorem  provides conditions\footnote{The result in~\eqref{eq:thmcondition} also holds for the case where the regularization term $N_0E_s^{-1}$ vanishes, which coincides to ZF detection. 
As a consequence, the condition~\eqref{eq:thmcondition} is rather pessimistic and is likely to be sub-optimal, especially for $N_0E_s^{-1}>0$. The derivation of a tighter condition is left for future work.} for which the Neumann series converges with a certain probability; this can be accomplished by setting $\alpha=1$ and $K=1$ and by inspecting  the convergence condition~\eqref{eq:convergencecriterion}.
 Furthermore, Theorem~\ref{theorem} \revision{provides conditions} for which the residual estimation error~\eqref{eq:errorbound1} is small. 
In both cases, we can see from Theorem \ref{theorem} that increasing the ratio between the number of BS antennas $\MR$ and the number of users $\MT$ increases the probability of convergence. Moreover, for $\alpha<1$, increasing~$K$ also increases the probability that the residual estimation error caused by a $K$-term approximation in \eqref{eq:errorbound1} is smaller than $\alpha$. 

\sloppy

We note that Theorem~\ref{theorem} also provides insight into the behavior in the large-antenna limit, i.e., for $\MR\to\infty$ while $\MT$ is held constant. In this case, we have 
$\mathrm{Pr}\!\left\{\|\bD_w^{-1}\bE_w\|^K_{F}< \alpha\right\}\to 1$ for $\alpha\in(0,1]$, which implies (i) that the Neumann series converges with probability~$1$ and (ii) that the approximation error for any $K$-term approximation is arbitrary small, which includes the MF detector (corresponding to~$K=1$). We note that this behavior is in accordance \revision{with} existing  results for MF detection in large-scale MIMO systems~\cite{Marzetta2010, hoydis2011massive}.

\fussy

\subsection{Channel Gain and NPI Variance Computation}
\label{sec:NPIwithapproximation}

Computation of the max-log LLRs via the proposed Neumann series approximation is carried out by simply replacing the exact inverse ${\bA}_w^{-1}$ by the approximation $\tilde{\bA}_{w\mid K}^{-1}$ to perform MMSE equalization~(\ref{eq:maxlogllr}). For this approximation, the effective channel gain $\tilde{\mu}^{(i)}_K$ and the variance of the residual post-equalization NPI variance $\tilde{\nu}_{i\mid K}^2$ now depend on the number of Neumann series terms.


In order to compute the effective channel gain~$\tilde{\mu}^{(i)}_K$, we first construct the $L\MT\times L\MT$ matrix~$\widetilde{\bW}^{-1}_{\cdot\mid K}$ from the sub-carrier equalization matrices~$\widetilde{\bW}_{w\mid K}^{-1}$, $\forall w$, as explained in Section~\ref{sec:lindetection}. With this, we have~$\tilde{\mu}^{(i)}_K x^{(i)}_t =\mathsf{E}\big[\bmf^{H}_t\widetilde{\bW}_K^{(i,:)}\bmy\mid x^{(i)}_t \big]$, which can be rewritten as in~(\ref{eq:mu}) by replacing $\bW^{(i,:)}$ with $\widetilde{\bW}_{\cdot\mid K}^{(i,:)}$. Consequently, the effective channel gain is given by $\tilde{\mu}^{(i)}_K= L^{-1}\sum_{w=1}^{L}\widetilde{\bmw}^H_{i,w\mid K}{\bmh}_{i,w}$, where  $\widetilde{\bmw}^H_{i,w \mid K}$  is the $i^\text{th}$ row of $\widetilde{\bW}_{w\mid K}^{(i,i)}$ and $\bmh_{i,w\mid K}$ the $i^\text{th}$ column of~$\bH_{w}^{(i,i)}$.


%
In order to compute the post-equalization NPI variance~$\tilde{\nu}_{i|K}^2$ one might assume that it simply corresponds to $E_s\tilde{\mu}_K^{(i)}-E_s|\tilde{\mu}_K^{(i)}|^2$ \revision{as in~\eqref{eq:approxexactllr}}.
Unfortunately, this expression no longer holds, \revision{because of the following fact:} 
\begin{align*}
\widetilde{\bW}_{\cdot\mid K}(E_s\bH\bH^{H}+N_0\bI_{L\MR})\neq\bH^H.
\end{align*} 
Furthermore, the above NPI variance expression is not guaranteed to be non-negative and hence, using it to compute LLR values inevitably results in poor error-rate performance.
As a consequence,  an alternative expression for $\tilde{\nu}_{i|K}^2$ is required when using the approximate matrix inverse for data detection. 
Following the steps of the derivation of $\tilde{\nu}_i^2$ in Section~\ref{subsubsection:exactllr} and by replacing $\bW^{(i,:)}$ with $\widetilde{\bW}_{\cdot\mid K}^{(i,:)}$, the exact post-equalization NPI variance can be expressed as:
\begin{align}
\tilde{\nu}_{i|K}^2 = & \bmf^{H}_t\widetilde{\bW}_{\cdot\mid K}^{(i,:)}(E_s\bH\bH^{H}+ \ldots \notag \\  
& \quad N_0\bI_{L\MR})(\widetilde{\bW}_{\cdot\mid K}^{(i,:)})^H\bmf_t- E_s\!\abs{\tilde{\mu}_K^{(i)}}^2.\label{eq:varneumannexact}
\end{align}
Since $\bH^H(E_s\bH\bH^{H}+N_0\bI_{L\MR}) = (E_s\bH^{H}\bH+N_0\bI_{L\MT})\bH^H$, we have: 
\begin{align*}
\tilde{\nu}_{i|K}^2 
			&=E_s\bmf^{H}_t(\widetilde{\bA}_{\cdot\mid K}^{-1})^{(i,:)}\bA\bG(\widetilde{\bA}_{\cdot\mid K}^{-1})^{(i,:)}\bmf_t- E_s\!\abs{\tilde{\mu}_K^{(i)}}^2.
\end{align*}
As $(\tilde{\bA}_{\cdot\mid K}^{-1})^{(i,i)}$ is diagonal, we can decompose the computation as the sum of per-subcarrier operations. To this end, let $\tilde\bma^{H}_{i,w|K}$ be the $i^\text{th}$ row of $\widetilde{\bA}_{w|K}^{-1}$, then:
\begin{align} \label{eq:exactNPIvariance}
\tilde{\nu}_{i|K}^2  &=E_s\sum_{w=1}^{L}\tilde\bma^{H}_{i,w|K}\bA_{w}\bG_{w}\tilde\bma_{i,w|K}
- E_s\!\abs{\tilde{\mu}_K^{(i)}}^2.
\end{align}
This expression, however, is computational intensive, as it involves the $L$ matrix multiplications, each requiring $O(\MT^3)$ operations. 
In order to reduce the complexity of computing $\tilde{\nu}_{i\mid K}^2$, we can use the $K=1$ term approximation NPI 
\begin{align} \label{eq:firstnpiapproximation}
\tilde{\nu}^2_{i|1}=E_s\sum_{w=1}^{L}(d^{(i,i)}_w)^{-2}\bma_{i,w}^H\bmg_{i,w} 
- E_s\!\abs{\tilde{\mu}_1^{(i)}}^2
\end{align}
as a substitute for $\tilde{v}^2_{i\mid K}$. Here,  $d^{(i,i)}_w$ is the $i^\text{th}$ diagonal entry of $\bD_w$, $\bma^{H}_{i,w}$ is the $i^\text{th}$ row of ${\bA}_w$, and $\bmg_{i,w}$ is the $i^\text{th}$ column of ${\bG}_w$.
%
This approximation requires low computational complexity as it involves only $L$ inner products, each requiring~$\MT$ operations. In addition, the larger $K$ is, the closer the approximate inversion in~(\ref{eq:approximation}) \revision{is} to the exact inverse (assuming the Neumann series converges). Hence, for $K>1$, the exact NPI variance would be lower than $\tilde{\nu}_{i\mid K}^2$, which reveals that~\eqref{eq:firstnpiapproximation} is a pessimistic approximation. 

We emphasize that we can further reduce the computational complexity of the NPI approximation in \eqref{eq:firstnpiapproximation}. In particular, let $a_w^{(i,j)}$ be the $i^\text{th}$ entry of the vector $\bma_{j,w}$ and $g_w^{(i,j)}$ be the $i^\text{th}$ entry of the vector $\bmg_{j,w}$. Since  $\bA_w=\bG_w+{N_0}{\Es^{\!-1}}\bI_{\MT}$, we have the following identity:
\begin{align*}
 (d^{(i,i)}_w)^{-2}\bma_{i,w}^H\bmg_{i,w} = (d^{(i,i)}_w)^{-2} a_w^{(i,i)}g_w^{(i,i)} +\!\! \sum_{j,i \neq j}(a_w^{(i,j)})^H g_w^{(i,j)}.
\end{align*}
Since (i) $a_w^{i,j} = g_w^{i,j}$, $\forall i \neq j$, (ii) $d^{(i,i)}_w=a^{(i,i)}_w$, and (iii) $d^{(i,i)}_w \gg a^{(i,j)}_w$ in the case where $\MT\ll\MR$, we can use the approximation $(d^{(i,i)}_w)^{-2}\bma_{i,w}^H\bmg_{i,w} \approx (d^{(i,i)}_w)^{-1} g_w^{(i,i)}$. Hence, we propose the following low-complexity NPI approximation:
\begin{align} \label{eq:npiapproximation}
\tilde{\nu}^2_{i} \approx E_s\sum_{w=1}^{L}(d^{(i,i)}_w)^{-1} g_w^{(i,i)} - E_s\!\abs{\tilde{\mu}_1^{(i)}}^2.
\end{align}
Note that our own simulations show that the low-complexity NPI approximation \eqref{eq:npiapproximation} performs well compared to the exact NPI variance \eqref{eq:exactNPIvariance}. For example, the performance loss caused by the approximation compared to \revision{the} exact NPI computation for $\MT=4$, $\MR=8$, and $K=3$ is less than $0.02$\,dB at \revision{a} BLER of $10^{-2}$ (cf.~Section~\ref{sec:errorrateperformance} for the simulation settings). 

\subsection{Simulation Results}\label{sec:sims}

We next demonstrate the advantages and limitations of the proposed approximate matrix inversion approach in terms of computational complexity and error-rate performance. \revision{To assess the error-rate performance for practically relevant antenna configurations, we note that the Samsung Full-Dimensional MIMO prototype~\cite{SamsungFDMIMO} consists of 64 BS antennas, whereas the massive MIMO research platform developed at Rice University~\cite{ArgosV2} currently consists of 96 BS antennas (with plans for larger array sizes). Hence, we focus our results on the following cases: $\MR=64$, $\MR=128$, and $\MR=256$.}
\subsubsection{Computational complexity}

\sloppy

To demonstrate that the proposed approximate inverse exhibits (often significantly) lower complexity than an exact inverse, we chose a Cholesky decomposition-based inverse as a reference (see Section~\ref{sec:referencecholesky} for algorithm details), as this method exhibits lower complexity compared to other inversion algorithms, including (but not limited to) direct matrix inversion, QR decomposition, or LU factorization \cite{GV96,ABurgThesis}.
The computational complexity \revision{(characterized by the sum of real-valued division\footnote{\revision{The number of divisions is not significant for the total operation count.}},  addition, and multiplication\footnote{\revision{To obtain the real-valued multiplication count,  we assumed four real-valued multiplications per one complex-valued multiplication. One could further reduce the number of the real-valued multiplications by using strength-reduction; this approach, however, maintains the trends observed in Fig.~\ref{fig:complexity}.}} operations)} of an exact Cholesky-based inverse scales with $O(\MT^3)$, whereas the complexity of a $K=1$ and $K=2$ Neumann series expansion scales only with $O(\MT)$ and $O(\MT^2)$, respectively.
\revision{The computational complexity of $K\geq3$ is dominated by matrix-by-matrix multiplications, where the number of such operations grows linearly with $K$. For example, $K=3$ requires one matrix-by-matrix multiplication, whereas $K=4$  requires two. In general, a $K\geq3$ term approximation requires $K-2$ matrix-by-matrix multiplications. As a result, the complexity of a $K\geq3$ term approximation is $O((K-2)U^3)$. 
Hence, we have $O(\MT^3)$ for $K=3$, which is equivalent to that of an exact Cholesky-based inverse.  Consequently,  a Neumann series approximation with $K\geq3$ does not appear to be advantageous. }

\fussy

\begin{figure}[t]
\centering
\includegraphics[width=0.95\columnwidth]{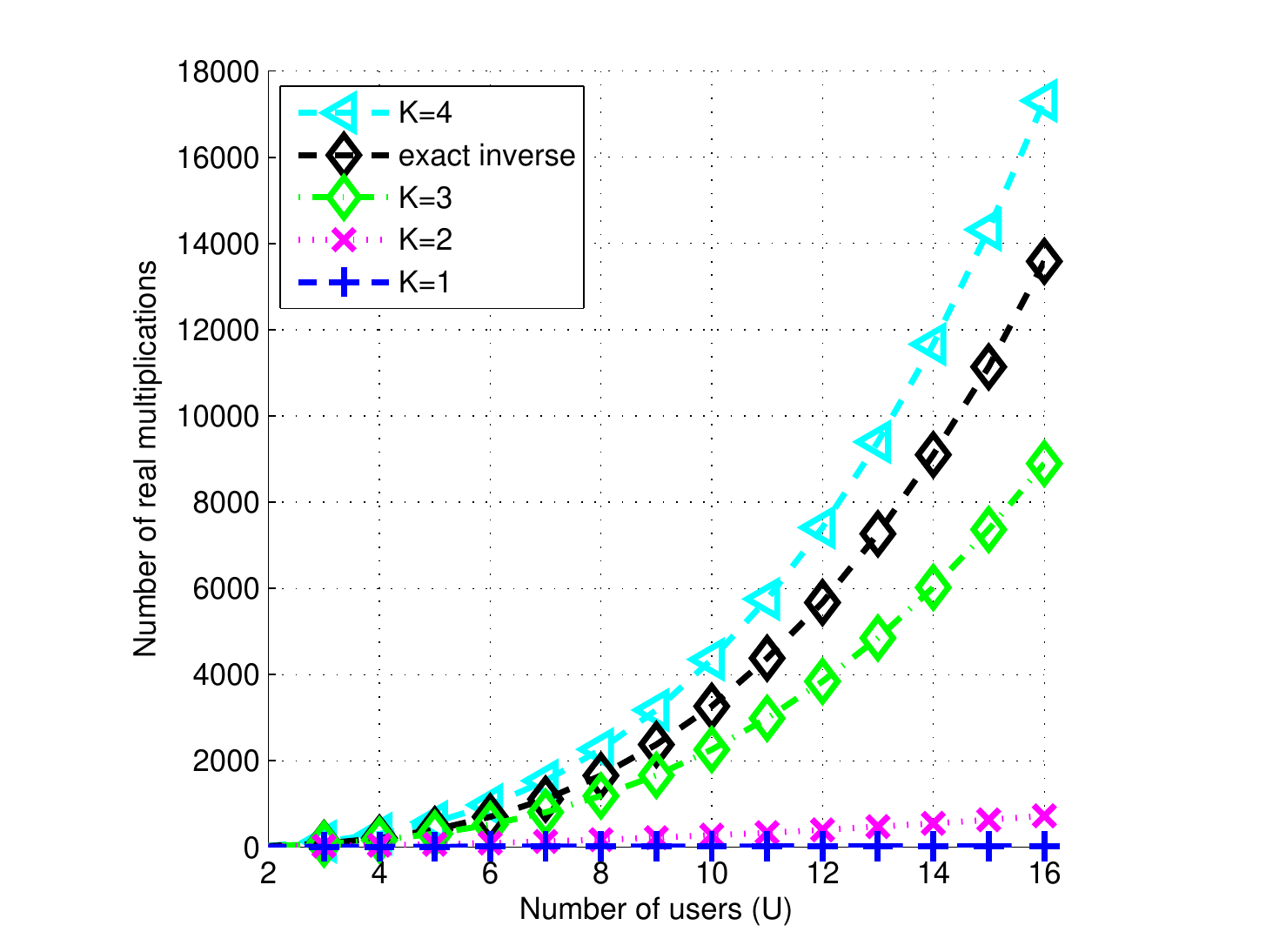}
\caption{Number of real-valued multiplications depending on the number of users~$\MT$. The proposed approximation with $K\leq3$ requires substantially lower complexity than that of an exact inverse based on the Cholesky decomposition.} \label{fig:complexity}
\end{figure}

\revision{The overall operation counts of both methods are dominated by the number of real-valued multiplications and additions, where real-valued multiplication is more expensive than real-valued addition. Since asymptotic complexity scalings do not, in general, reveal the full truth, we count the number of real-valued multiplications of both methods in  \figref{fig:complexity} for varying numbers of users $\MT$.}
\revision{We observe that for $K\leq3$, the Neumann series approach results in substantially lower complexity than the exact inversion approach. As expected, $K\geq4$ results in higher complexity than a Cholesky-based exact inversion.}

\subsubsection{Error-rate performance}
\label{sec:errorrateperformance}
Evidently, the reduction in complexity for $K\leq3$ Neumann series terms comes at the cost of an approximation error (cf.~Section~\ref{sec:approxerror}). 
To characterize the associated performance loss, we now compare the error-rate performance of the \revision{proposed} approximate matrix inverse \revision{with the error-rate performance of the exact inversion} for an LTE-based large-scale MIMO uplink system. To this end, we show simulation results of an SC-FDMA LTE uplink system with $\MR$ antennas at the BS and $\MT\leq \MR$ single-antenna users. 
In particular, we study a challenging communication scenario (from an error-rate perspective) and focus on the MCS (modulation and coding scheme) of \revision{the} highest rate (i.e., MCS 28) and  $20$~MHz bandwidth with $1200$ subcarriers,  as specified by the LTE standard~\cite{3GPPLTE}; this mode corresponds to $64$-QAM, and a rate~$\approx0.75$ 3GPP LTE turbo code.
In order to generate channel matrices that reflect a potential\footnote{To the best of our knowledge, no \revision{specific} channel model for large-scale MIMO systems is available in the open literature.} real-world scenario, we use the WINNER-\mbox{Phase-2} model~\cite{winner2}. In addition, we assume a linear antenna array with an antenna spacing of $10/128\approx0.0781$\,m, which \revision{resembles that of} the real-world channel measurement campaign in~\cite{HoydisChannel2012}. 
At the BS, we use the exact and approximate soft-output MMSE detectors detailed above. \revision{Furthermore, we} use a log-MAP LTE turbo decoder that performs~$16$ (full-)iterations.
\revision{Further, we define signal-to-noise-ratio (SNR) as $\MR E_s/N_0$, which corresponds to the average SNR per receive antenna.}

\begin{figure}[tp]
\centering
\subfigure[]{\includegraphics[width=0.9\columnwidth]{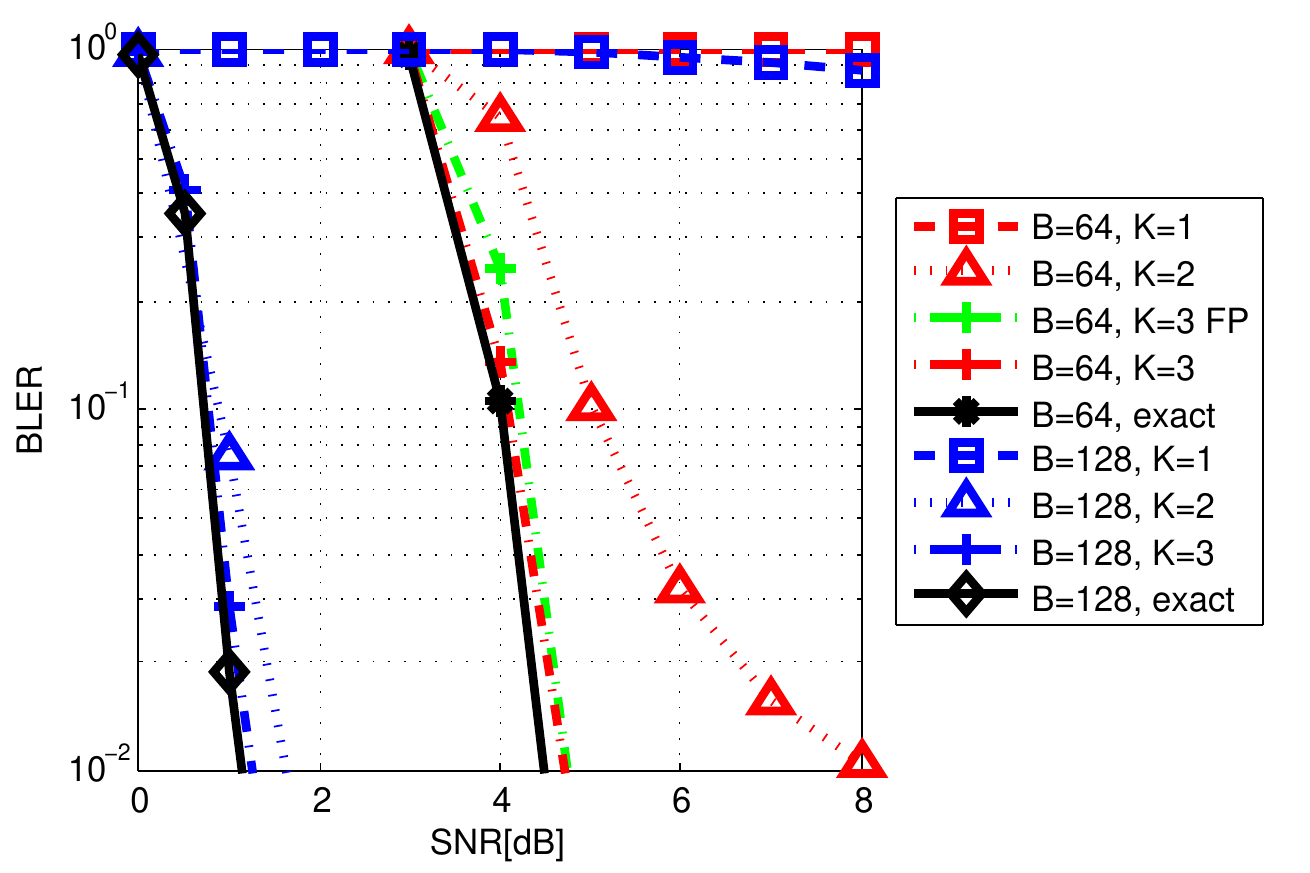}\label{fig:4_128_BLER}}
\subfigure[]{\includegraphics[width=0.9\columnwidth]{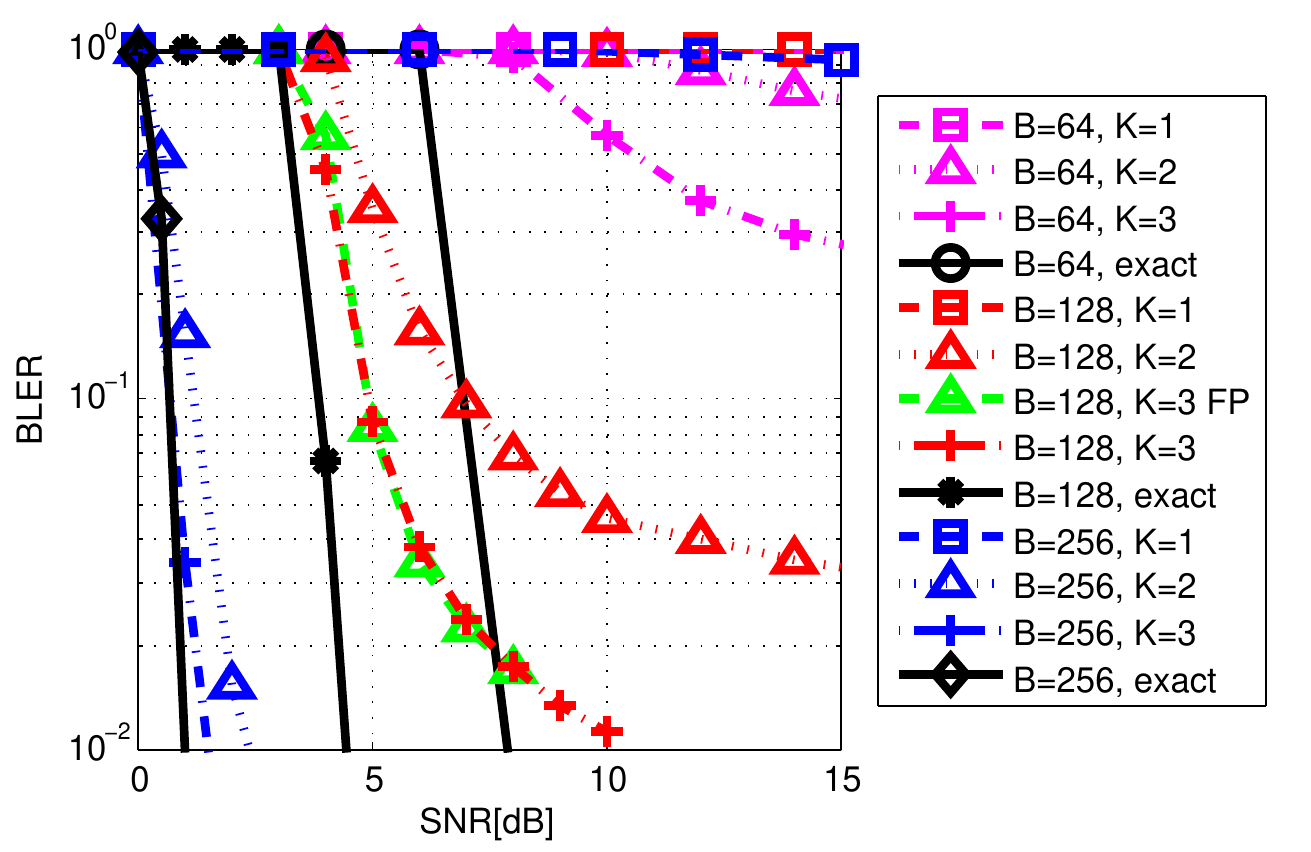}\label{fig:8_256_BLER}}
\subfigure[]{\includegraphics[width=0.9\columnwidth]{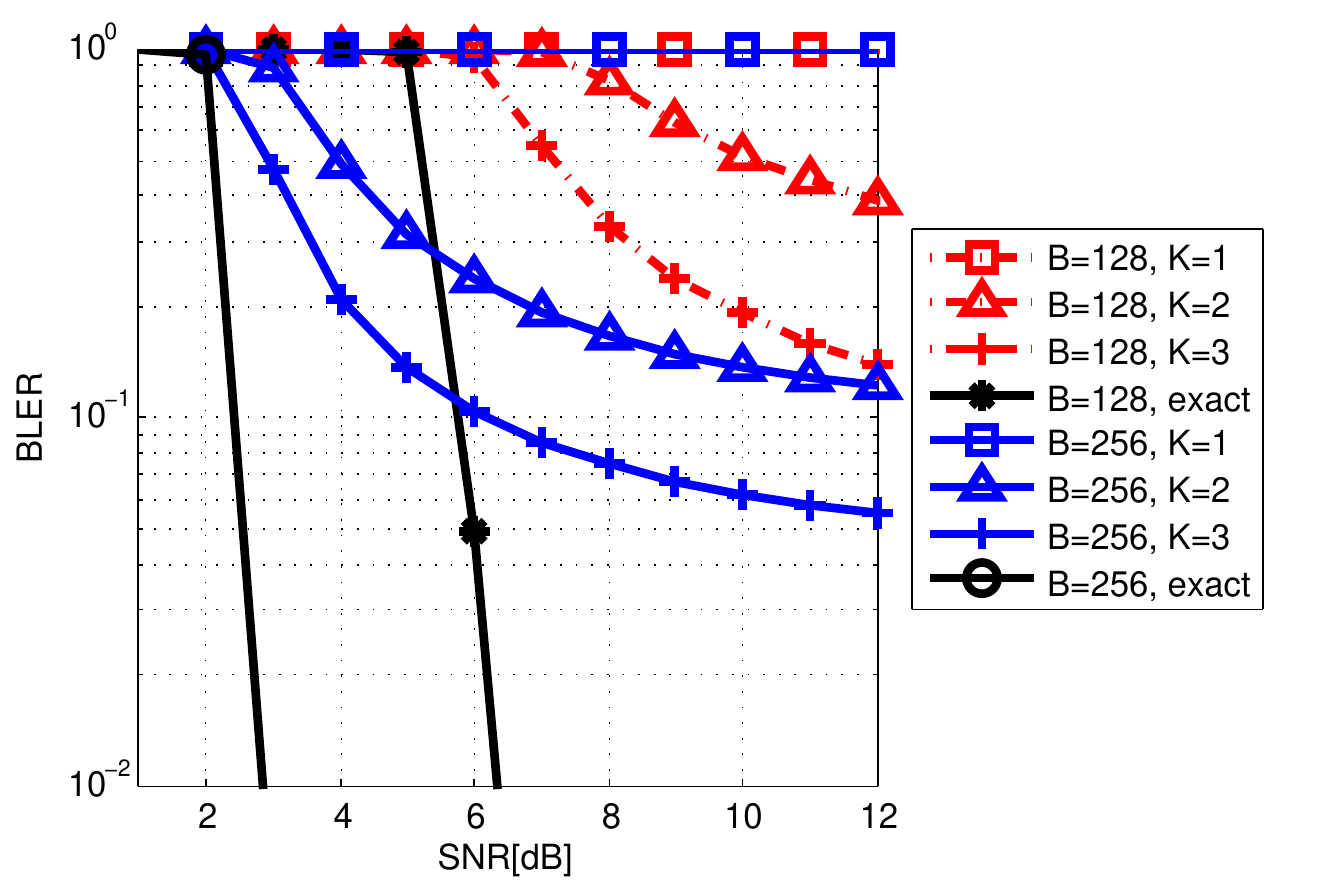}\label{fig:12_256_BLER}}
\caption{\revision{Block error-rate (BLER) performance comparison for (a) $\MT=4$ (b) $\MT=8$,  and (c) $\MT=12$ single-antenna users where $M=64$ and $\text{MCS}=28$; `FP' designates the performance of a fixed-point implementation.}}
\label{fig:figure} 
\end{figure}

\revision{Figures~\ref{fig:4_128_BLER}, \ref{fig:8_256_BLER}, and~\ref{fig:12_256_BLER}} show the block-error rate (BLER) performance of the proposed approximate detection algorithm compared to that of an exact MMSE detector for $\MT=4$, $\MT=8$ and $\MT=12$, respectively.  

\revision{We see that for small ratios between BS antennas and users, the MF detector (equivalent to $K=1$) and the Neumann series approximation for $K=2$ result in large residual errors.\footnote{Compared to lower modulation orders, such as $16$-QAM (not shown here), $64$-QAM requires a relatively high SNR to perform well.}}
Hence, considering the $10$\% BLER requirement for LTE~\cite{3GPPLTE}, the MF detector and $K=2$ term approximation \revision{are not suitable in practice} in the considered  \revision{$64$-QAM} cases (note that this fact is also reflected by Theorem~\ref{theorem}). 
{\revision{For a larger number of BS antennas, this error floor can be recovered partially. Our own simulations have shown that the MF detector achieves $<10^{-2}$ BLER for  $\MT=4$ and $\MR=512$. Furthermore, for $16$-QAM, our approximation method requires smaller values of~$K$ (see~\cite{Wu2012,Yin2013} for corresponding 16-QAM simulations in a large-scale MIMO-OFDM setting).}
}

\revision{We see that for $64$-QAM, the proposed approximate inversion method with $K=3$ terms is able to approach  the performance of the exact detector, i.e., the BLER performance loss is less than $0.25$\,dB SNR at $10^{-2}$ BLER in all of the $K=3, \MT=4$ cases and the $K=3, \MT=8, \MR=256$ case. Hence, the proposed approximate inverse for $K=3$ can deliver the performance of an exact inversion at (often substantially) lower complexity for large ratios between BS antennas and users. 
For small antenna ratios, however, the approximate inverse with $K=3$ exhibits an error floor. 
}

\revision{We conclude that systems with small ratios between BS and user antennas will need to resort to an exact inverse, while systems with large ratios can take advantage of the proposed approximate inverse. Hence, we next propose corresponding MIMO detection architectures for both, the approximate inverse and an exact Cholesky-based inverse.}


\section{VLSI Architecture}
\label{sec:architecture}

We now detail two VLSI architectures suitable for large-scale MIMO detection in 3GPP LTE-A. \revision{The first design implements the proposed approximate inversion approach and the second design implements an exact inverse; this enables us to perform a fair hardware complexity vs.\ error rate performance  comparison (see~Section~\ref{sec:implementation} for the comparison).}

\subsection{Architecture Overview}
\begin{figure*}[tp]
\centering
\includegraphics[width=1.7\columnwidth]{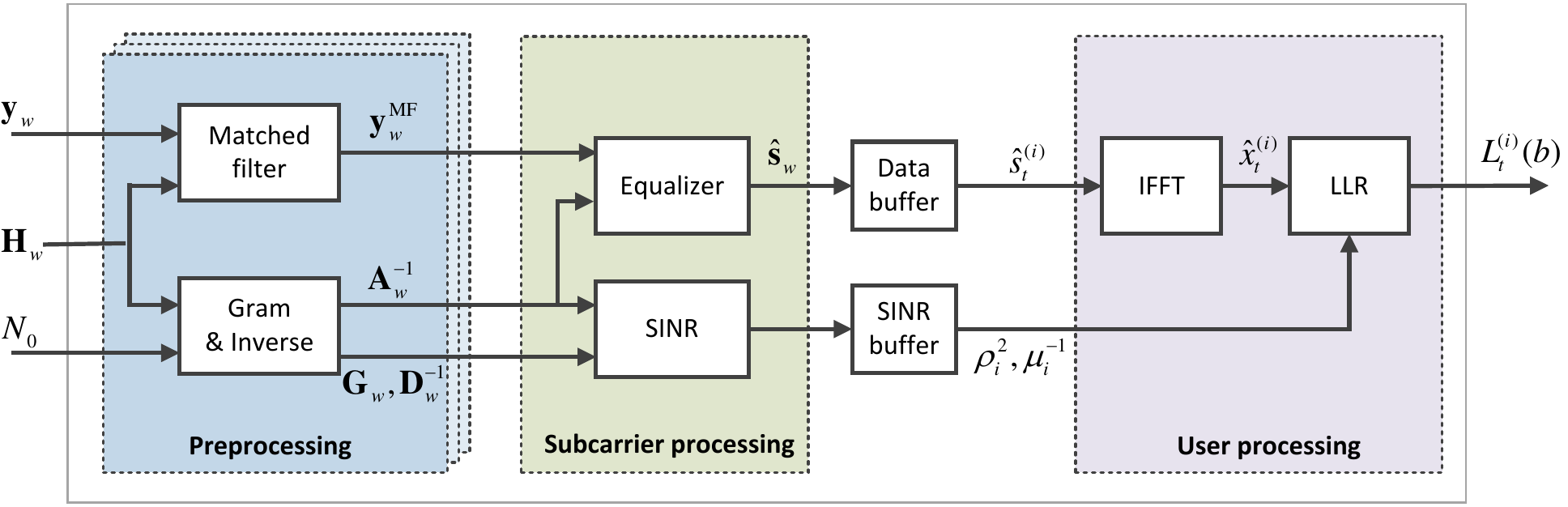}
\caption{High-level VLSI architecture of the large-scale MIMO detection engine for 3GPP LTE-A.}\label{fig:System}
\end{figure*}
\sloppy

The proposed general architecture is depicted in~\figref{fig:System} and consists of the following parts. The preprocessing unit performs matched filter computation, i.e., computes \mbox{$\bmy^\text{MF}_w=\bH_w^H\bmy_w$}, the regularized Gram matrix, and the (approximate) inverse. 
Note that for the approximate inversion unit, we also output $\bD_w^{-1}$ and $\bG_w$, which are needed to compute the SINR (cf.~Section~\ref{sec:NPIwithapproximation}). 
\revision{To achieve the peak throughput specified in LTE-A~\cite{3GPPLTEA}, while being able to handle the (worst) case where the channel estimates change from subcarrier to subcarrier and from SC-FDMA symbol to SC-FDMA symbol~(see, e.g,~\cite{Simko2011}), we use multiple instances of the preprocessing unit.\footnote{\revision{In many practical scenarios, the channel estimates may change only slowly. Hence, one does not need to  compute the inverse  for every SC-FDMA symbol. This fact could be either exploited to reduce the power consumption or to increase the achievable throughput of our detector designs.}}}
The matched filter output, the (approximate) inverse, and the regularized Gram matrix, 
are then passed to the subcarrier processing unit. 
This unit performs equalization, i.e., computes $\hat{\bms}_w=\bA^{-1}_w\bmy^\text{MF}_w$ and the post-equalization SINR (detailed in Section~\ref{subsubsection:exactllr} for the exact inverse and in  Section~\ref{sec:NPIwithapproximation} for the Neumann series approximation). 
\revision{To perform per-user data detection, a buffer is required that aggregates all equalized symbols and SINR values, which are computed on a per-subcarrier basis.}
The architecture then performs an IFFT, which transforms the equalized symbols from the subcarrier domain into the user domain (or time domain).
\revision{The LLR computation unit finally computes, together with the buffered post-equalization NPI values, soft-output information in the form of max-log LLRs~\eqref{eq:maxlogllr}.}
\revision{We next provide the details for the key blocks of the proposed detector architecture.}

\fussy

\subsection{Approximate Inversion and Matched Filter Units}
\label{sec:approxinvunit}

\subsubsection{Approximate inverse computation}

\begin{figure*}[tp]
\centering
\includegraphics[width=1.55\columnwidth]{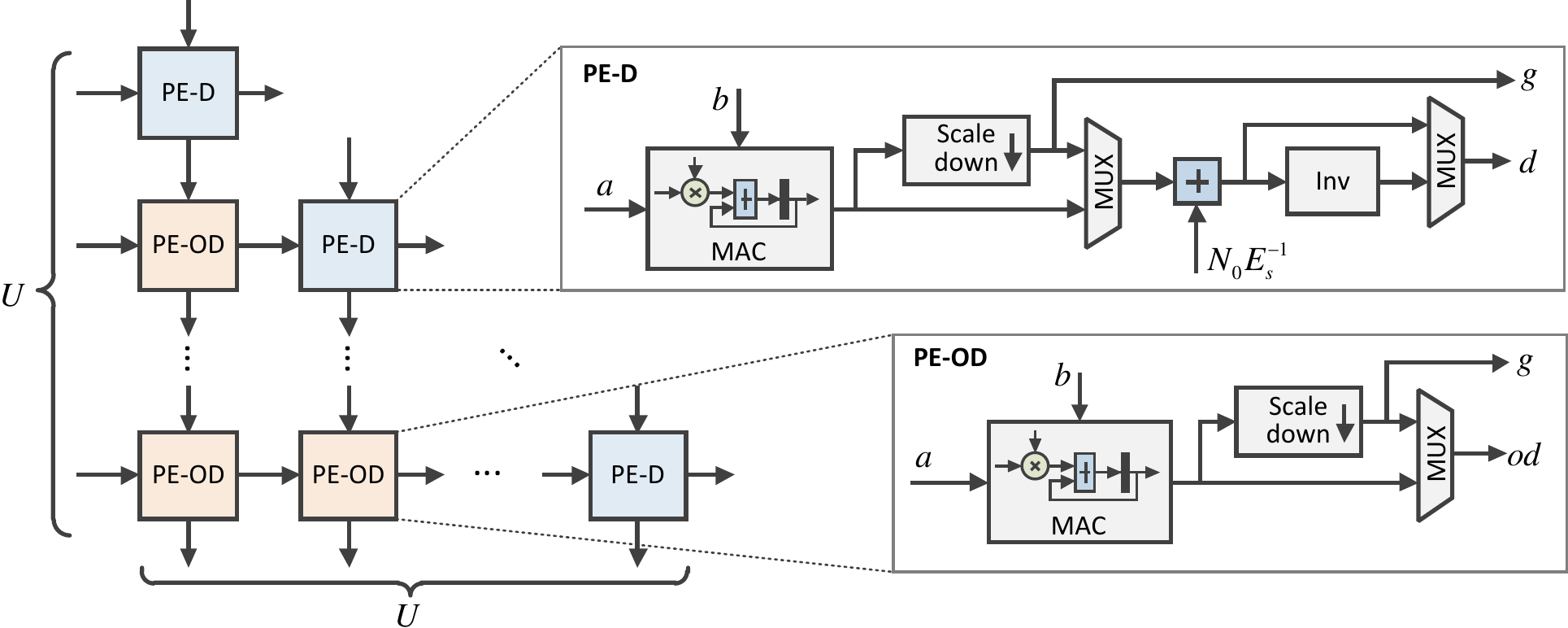}
\caption{Architecture details of the Gram matrix computation and approximate matrix inversion unit. The lower-triangular systolic array shown on the left consists of two processing elements (PEs); their architectural details are shown on the right.}\label{fig:sysarray}
\end{figure*}

 In order to achieve high throughput, we propose a \emph{single} systolic array that computes both, the regularized Gram matrix and the approximate inverse in four phases. The proposed architecture is detailed in~\figref{fig:sysarray} and \revision{is} capable of computing  inverses for various $K$-term expansions, i.e., the number of Neumann series terms can be selected at run-time. 
As shown in~\figref{fig:sysarray}, the lower triangular systolic array consists of two distinct processing elements (PEs): (i) PEs on the main diagonal of the systolic array (referred to as PE-D) and PEs on the off-diagonal (referred to as PE-OD).
As detailed next, both PEs have different modes in the four computation phases.
%

%
In the first phase,  the $\MT\times\MT$ normalized regularized Gram matrix $\bA_w/\MR=(\bG_w+N_0 E_s^{-1}\bI_U)/\MR$ is computed in $\MR$ clock cycles. Since $\bA_w$ is diagonally dominant with diagonal entries close to $\MR$, i.e., the number of BS antennas, we reduce its dynamic range by computing a normalized version, whose entries on the main diagonal are close to $1$ by  the `\textsf{scale down}'  unit shown in~\figref{fig:sysarray}; this trick mitigates dynamic-range issues, which are common for matrix inversion circuits implemented with fixed-point arithmetic.
The systolic array also computes  $\bD_w^{-1}\MR$ from the diagonal entries of $\bA_w/\MR$. These entries are computed in reciprocal units (denoted by `\textsf{inv}' in~\figref{fig:sysarray}) residing in the PE-D units.
The results~$\bD_w^{-1}\MR$ and $\bE_w/\MR$ are then stored in register files distributed in the systolic array. 

In the second phase, the  systolic array computes $-\bD_w^{-1}\bE_w$, by using the matrices $\bD_w^{-1}\MR$ and $\bE_w/\MR$ computed in the first phase. 
Since the matrix $-\bD_w^{-1}\bE_w$ is not Hermitian, the systolic array computes the upper- and lower-triangular parts of  $-\bD_w^{-1}\bE_w$ separately.
As  $\bD_w^{-1}$ is a diagonal matrix, computation of $-\bD_w^{-1}\bE_w$ only requires a series of scalar multiplications (rather than a matrix multiplication)\revision{.}

\sloppy

In the third phase, the systolic array computes the \mbox{$K=2$} term Neumann series approximation, i.e., $\widetilde{\bA}^{-1}_{w\mid 2}\MR=(\bD_w^{-1}\MR-\bD_w^{-1}\bE_w\bD_w^{-1}\MR)$. 
To this end, it is important to realize that the matrix $\bD_w^{-1}\MR-\bD_w^{-1}\bE_w\bD_w^{-1}\MR$ is Hermitian, implying that only the lower triangular part needs to be computed. Furthermore, since $\bD_w^{-1}\MR$ is diagonal, computation of   $-\bD_w^{-1}\bE_w\bD_w^{-1}\MR$ only requires entry-wise multiplications (instead of costly matrix multiplications). These scalar multiplications are carried out by loading $\bD_w^{-1}\MR$ and $-\bE_w\bD_w^{-1}$ into all PEs and performing a scalar multiplication to compute $\bD_w^{-1}\bE_w\bD_w^{-1}\MR$. Then, we add $\bD_w^{-1}\MR$ to the result in the diagonal PEs.
The result of this phase, i.e., $\bD_w^{-1}\MR-\bD_w^{-1}\bE_w\bD_w^{-1}\MR$, is stored in the distributed register files.

\fussy

In the fourth phase, the $K$-term Neumann series approximation is computed with the results residing in the distributed register files. In particular, the systolic array first performs a matrix multiplication of $-\bD_w^{-1}\bE_w$ with $\widetilde{\bA}^{-1}_{w\mid K-1}\MR$, and then adds $\bD_w^{-1}\MR$  to the diagonal PE. The resulting $K$-term approximation $\widetilde{\bA}^{-1}_{w\mid K}\MR$ is then stored in the register files.
This phase can be repeated for a configurable number of iterations, \revision{which allows us to compute an arbitrary $K$-term approximation.} 

\subsubsection{Matched filter computation}
The matched filter (MF) unit consists of a linear array of $\MT$ PEs. Each PE is associated \revision{with} one row of the Hermitian matrix $\bH_w^H$, and contains a single multiply accumulate unit (MAC) and a scaling unit to normalize the result to $\bmy^\text{MF}_w/\MR$. The MF unit reads a new entry of~$\bmy_w$ every clock cycle, and multiplies it with the corresponding entries in $\bH_w^H$ in each PE and \revision{then,} adds it the  previous results; the \revision{final} result is then normalized by $1/\MR$. 

\subsection{Equalization and SINR Computation Units}

\subsubsection{Equalization unit}
The equalization unit consists of a linear array of $\MT$ MAC units, and reads the normalized approximate inverse $\widetilde{\bA}^{-1}_{w\mid K}\MR$ and the $\bmy^\text{MF}_w/\MR$ from the matched filter unit. For each clock cycle, this unit  takes one column of $\widetilde{\bA}^{-1}_{w\mid K}\MR$, multiplies it with one element from  $\bmy^\text{MF}_w/\MR$,  \revision{and adds the scaled column to the previous results. The unit outputs an equalized symbol $\hat{\bms}_w$ every $\MT$ clock cycles.} 

\subsubsection{SINR computation unit}
The SINR computation unit simply consists of $\MT$ MAC units that sequentially \revision{compute} the approximate effective channel gain~$\tilde{\mu}^{(i)}_K$. This unit furthermore computes the approximate NPI \eqref{eq:npiapproximation} using a single MAC unit.
Subsequently, the unit multiplies $\tilde{\mu}^{(i)}_K$ with \revision{the} reciprocal of the approximate NPI $\tilde{\nu}^2_{i}$ to obtain the post-equalization SINR $\rho^{2}_i$. The same unit computes the reciprocal of  $\tilde{\mu}^{(i)}_K$ which is used in the LLR computation unit detailed next.

\subsection{IFFT and LLR Computation Units}

\subsubsection{IFFT unit}

In order to transform the per-subcarrier data into the user (or time) domain, we deploy a single Xilinx Discrete Fourier Transform IP LogiCORE unit (see \cite{DFTcorexilinx} for the specifications). This unit supports all forward and inverse DFT modes specified in 3GPP LTE~\cite{3GPPLTE}, but we only make use of its IDFT capabilities. The IFFT unit reads and outputs data in \revision{a} serial manner. For an IFFT transform size of $1200$ subcarriers, the core can process a new set of data every $3779$ clock cycles. This FFT unit \revision{achieves more than $317$\,MHz on a \revision{Virtex-7 XC7VX980T} FPGA and hence, achieves a throughput beyond $600$\,Mb/s for $8$ users, $64$-QAM, and 20MHz bandwidth}.

\subsubsection{LLR computation unit}

\sloppy

The LLR computation unit (LCU) generates max-log soft output values given the effective channel gains $\mu^{(i)}$ from the IFFT block and the post-equalization SINR values $\rho^{2}_i$ obtained from the SINR block. 
Since LTE specifies Gray mappings for all modulation schemes (BPSK, QPSK, 16-QAM, and 64-QAM), one can simplify the computation of the max-log LLR values in~\eqref{eq:maxlogllr} by rewriting \mbox{$L^{(i)}_{t}(b) = {\rho^{2}_i}\lambda_b(\hat{x}^{(i)}_t)$} and realizing that $\lambda_b(\cdot)$ is a piecewise linear function that depends on the bit index (see~\cite{Studer2011} for the details). 
To this end, the LCU first scales the real and imaginary \revision{parts} of the equalized time-domain symbol with the reciprocal of the effective channel gain $1/\mu^{(i)}$. Then, it evaluates the piecewise linear function $\lambda_b(\hat{x}^{(i)}_t)$ and scales the result with the post-equalization SINR $\rho^{2}_i$. The resulting max-log LLR value is then delivered to the output of the unit.
In order to minimize the circuit area, the proposed architecture evaluates each piecewise linear function with \revision{logical} shifts and additions only. The reciprocals are computed with a lookup table that is stored in B-RAM units (see {\cite{Wu2012} for architectural details). A single instance of the resulting LCU is able to processes one symbol every clock cycle, resulting in a peak throughput of $1.89$~Gb/s for 64-QAM at $317$~MHz.

\fussy

\subsection{Reference Cholesky-based Inversion Unit}
\label{sec:referencecholesky}

In order to enable a fair performance/complexity assessment of the proposed approximate matrix inversion unit, we also implemented a reference unit that performs an exact matrix inversion. This unit simply replaces the approximate inverse unit detailed in Section~\ref{sec:approxinvunit}.
We next summarize the used Cholesky-based inversion algorithm and then, outline the corresponding VLSI architecture. 

\subsubsection{Inversion algorithm}
\label{sec:efficientinversion}
In \revision{the} proposed exact inversion unit, we compute $\bA^{-1}_w$ in three steps: (i) we form the regularized Gram matrix $\bA_w=\bG_w+N_0 E_s^{-1}\bI_U$; (ii) we perform a Cholesky decomposition according to $\bA_w=\bL_w\bL_w^H$, where~$\bL_w$ is a lower-triangular matrix with real-values on the main diagonal \cite{GV96}; (iii) we compute the inverse $\bA_w^{-1}$ using an efficient forward/backward substitution procedure \revision{proposed} in~\cite{Studer2011}. Specifically, we first solve $\bL_w\bmu_i=\bme_i$ for $\bmu_i$, $i=1,\ldots,\MT$, where $\bme_i$ is the $i^\text{th}$ unit vector, via forward substitution. We then solve $\bL^H_w\bmv_i=\bmu_i$ for $\bmv_i$, $i=1,\ldots,\MT$, via back substitution, which leads to the desired inverse $\bA^{-1}_w=[\,\bmv_1 \cdots \bmv_{\MT}\,]$. Note that this approach avoids a costly matrix-by-matrix multiplication, which would be needed by directly computing \mbox{$\bA^{-1}_w=(\bL_w^{H})^{-1}\bL_w^{-1}$}.

\subsubsection{Cholesky decomposition architecture}
The VLSI architecture for the Cholesky-based inverse differs from the one in Section~\ref{sec:approxinvunit}. In particular, we deploy three separate  units that compute (i) the regularized Gram matrix, (ii) the exact inverse using the above algorithm, and (iii) a forward/backward substitution unit to compute the inverse $\bA_w^{-1}$. All units are detailed next and separated by pipeline stages.

The regularized Gram matrix is computed as a sum of outer products, i.e., as $\bG_w=\sum_{i=1}^{\MR}\bmr_i\bmr^H_i$, where $\bmr_i$ designates the $i^\text{th}$ row of $\bH_w$. Since the Gram matrix is symmetric, it can be computed efficiently with a triangular systolic array of multiply and accumulate units~(MACs), similar to the array detailed  in Section~\ref{sec:approxinvunit}.
The Gram computation unit reads one row of $\bH_w$ at a time and \revision{is able to}  output a Gram matrix every $\MR^\text{th}$ clock cycle. To obtain the regularized Gram matrix~$\bA_w$, we add $N_0 E_s^{-1}$ to the diagonal of $\bG_w$ in the final clock cycle.

We then perform the Cholesky decomposition of $\bA_w$ with a lower-triangular systolic array to obtain the lower-triangular matrix~$\bL_w$. 
The systolic array consists of two distinct processing elements (PEs): (i)  the PEs on the main diagonal and (ii) the PEs on the off-diagonal. 
The data flow is similar to the linear systolic array~(the ``obvious case'') proposed in~\cite{schreiber1986systolic}. The difference is that our design processes an incoming column of $\bA_w$ with multiple PEs, whereas an incoming column is processed with a single PE in \cite{schreiber1986systolic}. \revision{As a result, our design is able to achieve the peak throughput requirements of LTE-A.}
In our design, the pipeline of one column of PEs is $16$ stages deep and streams out one column of $\bL_w$ every clock cycle (after a latency of $16(\MT-1)$ clock cycles). Consequently, the achieved throughput corresponds to one Cholesky decomposition every~$\MT$ clock cycles.

\subsubsection{Forward/backward-substitution architecture}

\sloppy

The forward\slash{}backward substitution unit (FBSU) receives a lower-triangular
matrix $\mathbf{L}_{w}$ as input, and computes $\mathbf{A}_{w}^{-1}=(\mathbf{L}_{w}^{H})^{-1}\mathbf{L}_{w}^{-1}$ as outlined in Section~\ref{sec:efficientinversion}.
The FBSU consists of three major components: (i) a forward substitution
unit (FSU), which solves for $\mathbf{L}_{w}\bmu_{i}=\mathbf{e}_{i}$,
(ii) a backward substitution unit (BSU), which solves for $\mathbf{L}_{w}^{H}\bmv_{i}=\bmu_{i}$, 
and (iii) a Hermitian transpose unit, which computes $\mathbf{L}_{w}^{H}$. 
Since the computations for the FSU and the BSU are symmetric,
we implement \revision{the} forward substitution architecture and re-use it for \revision{the} backward
substitution, by reversing the order of \revision{the} columns of the matrix $\mathbf{L}_{w}^{H}$
and vector $\bmu_i$ before reading them into the BSU. 
To simplify notation, we assume that the equation to be solved
by the forward substitution corresponds to $\mathbf{Lx=b}$ for some $\mathbf{x}$ and $\mathbf{b}$. Since the forward
substitution of solving the equation $\mathbf{L}\mathbf{x}_{i}=\mathbf{b}_{i}$
for each $\mathbf{b}_{i}$ $(i=1,\ldots,U)$ is independent, we use~$U$ processor elements (PEs) to solve for all $\mathbf{x}_{i}$ in parallel.
Each PE is implemented using a fully pipelined architecture, which
consists of $U$ stages of computation logic. Each stage contains
two multiplexers, a complex-valued multiplier, and a complex-valued subtraction.
In each stage, either $\Delta_{i}={b}_{i}-\sum_{j}{L}_{i,j}{x}_{j}$
or $\Delta_{i}/{L}_{i,i}$ is computed according \revision{to} the control signals. Therefore, for an input
matrix $\mathbf{L}_{w}$ of dimension $U$, the FSU uses $U^{2}$ complex-valued multipliers; the entire FBSU utilizes $2U^{2}$ complex-valued multipliers.
The matrix conjugate unit is implemented using multiplexers and $U$
FIFOs (realized by on-chip B-RAMs in the FPGA). The conjugate
matrix $\mathbf{L}_{w}^{H}$ is also reordered based on the pattern of
the input sequence of the BSU.

\fussy

\section{Implementation Results and Trade-offs}
\label{sec:implementation}

The approximate detection engine for 3GPP-LTE and the exact Cholesky-based detector have been implemented on a Xilinx Virtex-7 XC7VX980T FPGA. The fixed-point parameters, FPGA implementation results,  and the associated performance/complexity trade-offs are \revision{presented} next.

\subsection{Fixed-Point Design Parameters}

In order to minimize the hardware complexity, fixed-point arithmetic is used in the entire design. The associated fixed-point parameters were determined via extensive simulations. In the following, the word-lengths refer to the real or imaginary part of a complex-valued number. 

The channel matrices $\bH_w$, the receive-vectors $\bmy_w$, and the noise variance $N_0 E_s^{-1}$, are all quantized to  $15$\,bit.
The word-length of the output of the Gram matrix and inversion unit are also set to $15$\,bit; equivalently, the matched filter unit has $15$\,bit at the input and output. 
For both matrix inversion circuits, all multiplications have been mapped onto Xilinx DSP48 slices. 
In order to achieve sufficient precision at minimum implementation complexity, the MAC registers within the DSP48 units are set to $22$\,bit. 
The LUT in the reciprocal unit consists of $1024$ addresses with $12$\,bit outputs. Hence, it can be implemented efficiently using a single block-RAM (B-RAM) available on the FPGA. 
The equalizer module uses a $15$\,bit input and its output, which is stored in the data buffer, is quantized to $12$\,bit.
The buffer stores (complex-valued) data for $1200$ subcarriers and $\MT$ users. 
The SINR computation module has a $15$\,bit input and $12$\,bit output.
The input and output of the IFFT unit are $12$\,bit; the precision of the internal multipliers is set to $18$\,bit.
The inputs of the LLR computation are \revision{quantized} to $12$\,bit and the computed LLRs are represented by $8$\,bit. 

The resulting fixed-point performance is shown in \figref{fig:figure} (labeled by `FP') for $64\times4$ and $128\times8$ systems. As it can be seen, the fixed-point implementation is virtually indistinguishable from the floating-point golden model. In particular, the implementation loss is less than $0.05$\,dB SNR at 10\% BLER.

\subsection{FPGA Implementation Results}

\begin{table}
\begin{minipage}[c]{1\columnwidth} 
\centering
\caption{Implementation results on a Xilinx Virtex-7 XC7VX980T FPGA}
 \label{tbl:implresults}
\begin{tabular}{lcccc}
\toprule 
{Antenna configuration\footnote{$128\times8$ refers to $\MR=128$ BS antennas and $\MT=8$ single-antenna users.}} & \multicolumn{2}{c}{$128\times8$} & \multicolumn{2}{c}{\textbf{${64\times4}$}}\tabularnewline
{Inversion algorithm\footnote{$K=3$ designates the approximate inversion with $3$ Neumann series terms.}} & ${K=3}$ & {Cholesky} & ${K=3}$ & {Cholesky}\tabularnewline
\midrule 
{Clock frequency {[}MHz{]}} & 317 & 317 & 317 & 317\tabularnewline
{Throughput {[}Mb/s{]}} & 603 & 603 & 301 & 301\tabularnewline
\midrule 
\multirow{2}{*}{{LUT slices}} & 168\,125  & 208\,161  & 34\,631 & 78\,756 \tabularnewline
 & (28\%) & (34\%) &  (6\%) & (12.9\%)\tabularnewline
\midrule 
\multirow{2}{*}{{FF slices}} & 193\,451  & 213\,226 & 39\,492 & 39\,602\tabularnewline
 & (16\%) & (17.4\%) & (3.2\%) & (3.2\%)\tabularnewline
\midrule 
\multirow{2}{*}{{DSP48 units}} & 1\,059 & 1\,447 & 233 & 329\tabularnewline
 & (30\%) & (40.2\%) & (7\%) & (9.14\%)\tabularnewline
\midrule 
\multirow{2}{*}{{Block RAMs}} & 18 & 65 & 12 & 32\tabularnewline
 & (0.6\%) & (2.17\%) & (0.4\%) & (1.07\%)\tabularnewline
\bottomrule 
\end{tabular}
\end{minipage}
\end{table}

Table~\ref{tbl:implresults} summarizes the key (post-place-and-route) implementation results of the proposed approximate and exact soft-output data detector for LTE-based massive MIMO wireless systems.
We parameterized the architecture for $\MT$ and $\MR$ to explore the impact on the required FPGA resources and the  corresponding throughput. The implementation results for antenna configurations of $128\times8$ and $64\times4$ are detailed in Table~\ref{tbl:implresults}.
In order to support $75$\,Mb/s data rate for each LTE-A user in 20\,MHz bandwidth, we use multiple instances of the preprocessing unit. Specifically, we used $8$ and $5$ instances of approximate matrix inversion units for the $128\times8$ and $64\times4$  system, respectively. For the exact inverse, we used $6$ and $3$ regularized Gram matrix units for  the $128\times8$ and $64\times4$  system, respectively. In addition, we used one Cholesky decomposition unit and one forward and backward substitution unit for both cases to meet the data rate requirements.

As shown in Table~\ref{tbl:implresults},  all designs are capable of running at $317$\,MHz and the critical path is the routing between different blocks of the detector. For the $128\times8$ and $64\times4$ \revision{systems}, the proposed units can achieve $603$\,Mb/s and $301$\,Mb/s, respectively. For the $64\times4$ system, the design meets the $300$\,Mb/s peak data rate requirement specified in LTE-A with $4$ users and 20MHz bandwidth. In addition, our design can scale beyond LTE-A specifications, i.e., the proposed designs can support up to $8$ users and still achieve a $75$\,Mb/s per-user requirement.

In \revision{terms} used resources \revision{on the Virtex-7 XC7VX980T FPGA}, the approximate soft-output data detector is smaller than the Cholesky-based unit. There are notable saving in logic slices and DSP48 units. For $64\times4$, $K=3$ uses 56\% fewer LUT slices and 29\% fewer DSP48 units compared to that of the Cholesky-based unit. For $128\times8$, $K=3$ uses 19\% fewer LUT slices and 26\% fewer DSP48 units compared to that of  the Cholesky-based unit. We emphasize that the savings in hardware resources  become significantly larger as the number of users~$\MT$ increases.

\subsection{Performance/Complexity Trade-off}
\label{sec:tradeoffanaalysis}

Based on the simulated BLER results in~\figref{fig:figure} and the associated  FPGA implementation results, we are now ready to characterize the error-rate performance vs.~hardware complexity trade-offs associated with the detector containing the proposed approximate matrix inversion and the Cholesky-based exact inversion. 
To this end, we show  the associated hardware complexity against the minimum SNR required to achieve 10\% BLER in~\figref{fig:tradeoff}.
Since both designs are dominated by multipliers, we define the hardware complexity as the number of  multipliers  required to achieve a $75$\,Mb/s per-user throughput.

\begin{figure}[tp]
\centering
\includegraphics[width=0.75\columnwidth]{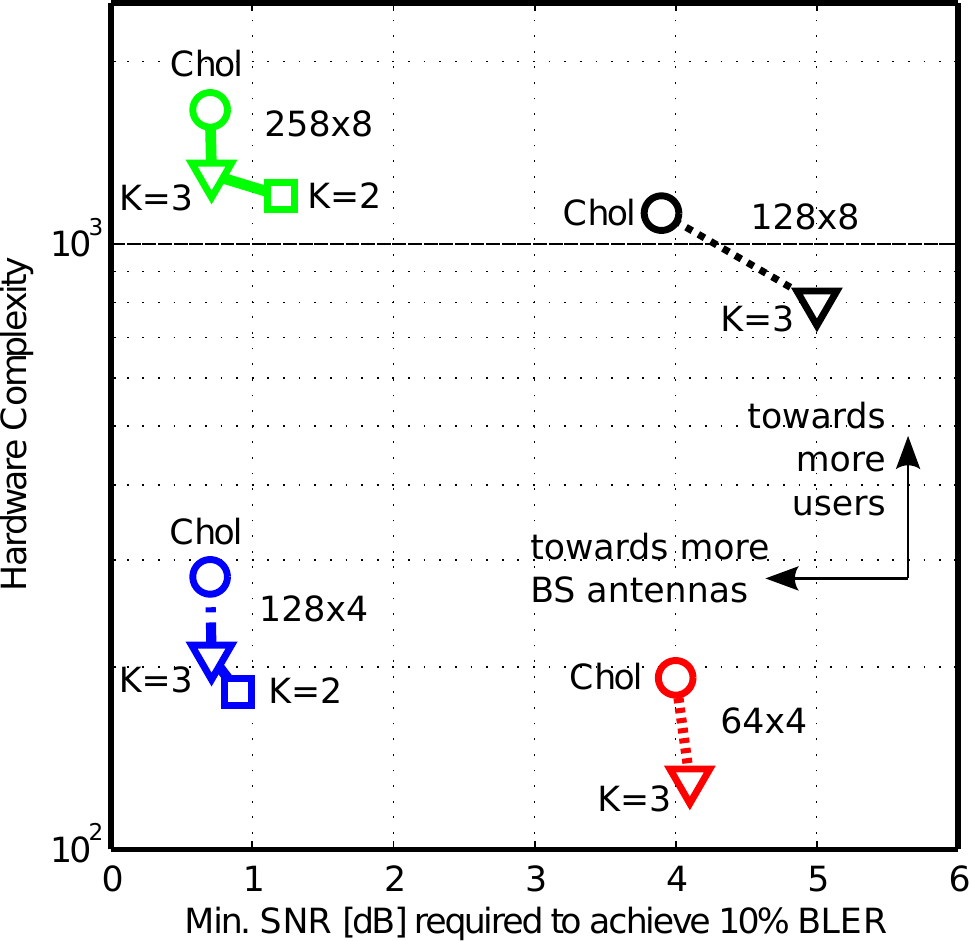}
\caption{Performance/complexity trade-off. Hardware complexity is defined as the number of DSP48E1 slices required to achieve the LTE-A uplink $75$\,Mb/s per-user peak throughput.}\label{fig:tradeoff}
\end{figure}

From~\figref{fig:tradeoff}, we \revision{observe} that the hardware complexity of the Cholesky-based detector is larger than that of the approximate inversion circuit for $K=3$ and $K=2$.
In addition, for large ratios between the number of BS antennas to the number of users $\MR/\MT$, we clearly see that the SNR performance of the approximate inverse with $K=3$ and the exact inverse are very similar.
For small ratios $\MR/\MT$, however, the performance difference between the approximate inverse and the exact inverse is rather large, which is reflected in the analysis shown Section~\ref{sec:approxerror}.
Hence, the ratio $\MR/\MT$ determines whether an approximate or exact inversion is beneficial in a practical large-scale MIMO system. 
Note that for $128\times 8$ and $64\times 4$, the approximate inverse with $K=2$ is unable to achieve 10\% BLER (cf.~\figref{fig:figure}). We \revision{note that when considering 16-QAM modulation (rather than 64-QAM modulation, as shown here), the approximate inversion for $K=2$ is capable of achieving similar performance \revision{as} the exact inverse (see~\cite{Wu2012,Yin2013} for corresponding simulation results).}

\subsection{Related FPGA Designs for Linear Data Detection}

A host of FPGA designs for linear data detection in conventional (small-scale) MIMO systems have been proposed in the literature~\cite{myllyla2005complexity, rao2010low, eberli2008divide,burg2006algorithm,WuDi2011,karkooti2005fpga,Cong2011,luethi2008gram}. Unfortunately, all these designs differ in various ways. First, the corresponding architectures rely on different matrix inversion algorithms, such as the QR decomposition~\cite{myllyla2005complexity, rao2010low, karkooti2005fpga, Cong2011}, Gram-Schmidt orthogonalization~\cite{kim2007efficient,luethi2008gram}, LU decomposition\cite{Studer2011}, direct matrix inversion~\cite{WuDi2011}, divide-and-conquer methods~\cite{eberli2008divide, Eilert2008}. Second, all FPGA \revision{implementations} do not generate soft outputs, with the exception of~\cite{wu2010vlsi}. 
Third, the designs were implemented on different FPGA types. 

Since the soft-output detector implementations proposed in this paper are for large-scale MIMO systems having hundreds of BS antennas and none of the small-scale MIMO detector designs in~\cite{myllyla2005complexity,rao2010low,burg2006algorithm,karkooti2005fpga,Cong2011,kim2007efficient,WuDi2011,eberli2008divide, Eilert2008} was implemented on a Xilinx Virtex-7 FPGA, a fair comparison of our design with the above-mentioned implementations is difficult. Hence, we decided \revision{to} resort to the comparison with our own reference circuit\revision{, i.e., the Cholesky-based inverse,} as shown in Section~\ref{sec:tradeoffanaalysis}.


\section{Conclusions}
\label{sec:conclusions}

\sloppy

We have proposed a new soft-output data detector for large-scale (or massive) MIMO-based 3GPP LTE-Advanced (\mbox{LTE-A}) systems. The proposed solution is capable of performing high throughput detection in single-carrier frequency division multiple access (SC-FDMA)-based large-scale MIMO systems equipped with hundreds of antennas at the base station~(BS). 
In order to achieve low computational complexity, we have proposed a new approximate linear detector relying on a Neumann series approximation of the matrix inverse. 
We have designed two reference VLSI architectures, one relying on the approximate inverse, the other on an exact Cholesky-based matrix inversion.
Both architectures have been successfully implemented on a state-of-the-art Xilinx \mbox{Virtex-7} FPGA, are suitable for systems equipped with $128$ BS antennas or fewer while serving up to $8$ users, and achieve more than $600$\,Mb/s, exceeding the peak data rates specified in the 3GPP LTE-A uplink for 20\,MHz bandwidth.
Our FPGA implementation results reveal that for systems with a large ratio between the number of BS antennas and the number of users, the approximate matrix inversion is able to significantly reduce the hardware implementation complexity (compared to that of the exact inversion) with only a slight error-rate performance degradation.
For systems with small ratios between the number of BS antennas and the number of users (as it is the case in, e.g., conventional, small-scale MIMO systems) one must resort to an exact inverse in order to avoid poor error-rate performance. This behavior is in accordance with the analytical results we have developed for the approximate matrix inverse. 
In summary, our FPGA implementation results demonstrate the practical feasibility of high-throughput data detection for 3GPP LTE-based large-scale MIMO systems. 
We finally note that a corresponding high-throughput ASIC design has recently been published in~\cite{Yin2014}.

\fussy

\revision{There are many avenues for future work. The development of detection algorithms that are able to perform iterative detection and decoding (as, e.g.,  in \cite{Studer2011}) in large-scale MIMO systems is left for future work. Furthermore, the design of high-performance, near-optimal detection methods (e.g., based on the algorithms in  \cite{HB03,SB10}) that require low computational complexity  for large-dimensional antenna configurations and for SC-FDMA is a challenging open research problem.}

\appendices


\section{Proof of Theorem \ref{theorem}}
\label{app:theoremproof}

To prove Theorem~\ref{theorem}, we need the following three Lemmata.
\begin{lem}
\label{lem1}
Let the scalars $x^{(k)}$ and $y^{(k)}$ for $k=1,\ldots,B$ be i.i.d.\ circularly symmetric complex Gaussian with unit variance. Then, $\mathsf{E}\!\left[\left|\sum_{k=1}^{\MR}x^{(k)}y^{(k)}\right|^4\right]=2\MR(\MR+1)$.
\end{lem}
\begin{proof}
We have 
\begin{align*}
& \mathsf{E}\!\left[\left|\sum_{k=1}^{\MR}x^{(k)}y^{(k)}\right|^4\right] 
\!\!=\!\mathsf{E}\!\left[\left(\sum_{k=1}^{\MR}x^{(k)}y^{(k)}\sum_{k=1}^{\MR}\left(x^{(k)}y^{(k)}\right)^*\right)^2\right]\\
& \qquad ={\MR\choose2}\mathsf{E}\!\left[|x^{(k)}|^2|y^{(k)}|^2\right]+4\mathsf{E}\!\left[|x^{(k)}|^4|y^{(k)}|^4\right]\\
& \qquad =2 \MR(\MR-1)+4 \MR=2\MR^2+2 \MR. 
\end{align*}
The above steps  can be summarized as follows. After expanding the quadratic expression, the non-zero terms can be written as $|x^{(k)}|^4|y^{(k)}|^4$ and $|x^{(k)}|^2|y^{(k)}|^2$, where $k=1,\ldots,\MR$. Then, there are $\MR$ terms of the form $|x^{(k)}|^4|y^{(k)}|^4$ and ${\MR\choose2}$ of the form $|x^{(k)}|^2|y^{(k)}|^2$. The facts that $\mathsf{E}\left[|x^{(k)}|^4\right]=\mathsf{E}\left[|y^{(k)}|^4\right]=2$ and $\mathsf{E}\left[|x^{(k)}|^2\right]=\mathsf{E}\left[|y^{(k)}|^2\right]=1$ concludes the proof.
\end{proof}
\begin{lem}
\label{lem2}
Let $\MR>4$ and $x^{(k)}$, $k=1,\ldots,\MR$ be i.i.d.\ circularly symmetric complex Gaussian with unit variance  and $g = \sum_{k=1}^{\MR}\mid x^{(k)} \mid^2$. Then, 
\begin{align} \label{eq:lemma2}
\mathsf{E}\!\left[\left|g^{-1}\right|^4\right]= \left((\MR-1)(\MR-2)(\MR-3)(\MR-4)\right)^{-1}.
\end{align}
\end{lem}
\begin{proof}
We first rewrite $g$ as $2^{-1}\sum_{k=1}^{2\MR}\mid s^{(k)} \mid^2$ where $s^{(k)}$, $k=1,\ldots,2\MR$, are i.i.d.\ zero-mean real-valued Gaussian with unit variance. Then, $2g^{-1}$ is \revision{an} inverse chi-square random variable with $2\MR$ degrees of freedom. The inverse chi-square distribution with $2\MR$ degrees of freedom $\chi(2\MR)$ corresponds to an inverse-Gamma distribution with $2B$ degrees-of-freedom. 
The $4^\text{th}$ moment of this inverse chi-square distribution is given by $\frac{1}{16}(\MR-1)(\MR-2)(\MR-3)(\MR-4)$~\cite{cook2008notes} and, hence, we obtain~\eqref{eq:lemma2}.
\end{proof}
\begin{lem}
\label{lem3}

Let $\MR>4$ and  the entries of $\bH_w\in\mathbb{C}^{\MR\times\MT}$ be i.i.d.\ circularly symmetric complex Gaussian with unit variance. Then, we have
\begin{align*}
&\mathsf{E}\left[\|\bD_w^{-1}\bE_w\|^2_{F}\right] \\ 
& \qquad \leq\left(\MT^2-\MT\right)\sqrt{\frac{2\MR(\MR+1)}{(\MR-1)(\MR-2)(\MR-3)(\MR-4)}}
\end{align*}
\end{lem}
\begin{proof}
The regularized Gram matrix corresponds to $\bA_w= \bD_w+\bE_w = \bG_w+N_0{E_s}^{-1}\bI_{\MT\times\MT}$. Thus, each element on the $i^\text{th}$ row and $j^\text{th}$ column of $\bA_w$, $a_w^{(i,j)}$ can be written as:
\begin{align*}
a_w^{(i,j)}\!=\!
\left\{\begin{array}{ll}
 	\!\!g_w^{(i,j)}\!=\!\sum_{k=1}^{\MR}\left(h_w^{(k,i)}\right)^*h_w^{(k,j)}, & \!\!\! i \neq j\\
	\!\!g_w^{(i,i)}\!+\!{N_0}E_s^{-1}=\sum_{k=1}^{\MR}\left|h_w^{(k,i)}\right|^2\!+\!{N_0}E_s^{-1}, & \!\!\! i = j,
\end{array}\right.
\end{align*}
with $g_w^{(i,j)}$ corresponding to the $i^\text{th}$ row and $j^\text{th}$ column of the Gram matrix $\bG_w$. 
We now have the following inequality:
\begin{align*}
\mathsf{E}\left[\|\bD_w^{-1}\bE_w\|^2_{F}\right]
&= \mathsf{E}\left[\sum_{i=1}^{i=\MT}\sum_{j=1, i \neq j}^{j=\MT}\left|\frac{g_w^{(i,j)}}{a_w^{(i,i)}}\right|^2\right] \\
& \leq \sum_{i=1}^{i=\MT}\sum_{j=1, i \neq j}^{j=\MT}\mathsf{E}
\left[\left|\frac{g_w^{(i,j)}}{g_w^{(i,i)}}\right|^2\right],
\end{align*}
which is obtained by omitting the non-negative regularization term $N_0E_s^{-1}$.
By applying the Cauchy-Schwarz inequality, we can bound $\mathsf{E}\left[\|\bD_w^{-1}\bE_w\|^2_{F}\right]$ from above as 
\begin{align*}
\mathsf{E}\!\left[\|\bD_w^{-1}\bE_w\|^2_{F}\right] &\leq \sum_{i=1}^{i=\MT}\sum_{j=1, i \neq j}^{j=\MT}
\!\!\sqrt{
\mathsf{E}\!\left[\left|g_w^{(i,j)}\right|^4\right]
\mathsf{E}\!\left[\left|\left(g_w^{(i,i)}\right)^{-1}\right|^4\right]}.
\end{align*}
Application of Lemmata~\ref{lem1} and~\ref{lem2} to the first and second expected values, respectively, we obtain
\begin{align*}
\mathsf{E}\left[\|\bD_w^{-1}\bE_w\|^2_{F}\right] &\!\leq\! \sum_{i=1}^{i=\MT}\sum_{j=1, i \neq j}^{j=\MT}
\!\!\sqrt{\frac{2\MR(\MR\!+\!1)}{(\MR\!-\!1)(\MR\!-\!2)(\MR\!-\!3)(\MR\!-\!4)} } \\
 &=\left(\MT^2\!-\!\MT\right)\sqrt{\frac{2\MR(\MR\!+\!1)}{(\MR\!-\!1)(\MR\!-\!2)(\MR\!-\!3)(\MR\!-\!4)} }.
\end{align*}
\end{proof}

We are now in position to prove Theorem \ref{theorem}. 
To this end, we start by using Markov's inequality to obtain the following straightforward inequality:
\begin{align*}
\mathrm{Pr}\!\left\{\|\bD_w^{-1}\bE_w\|^K_{F}\geq\alpha\right\} &= \mathrm{Pr}\!\left\{\|\bD_w^{-1}\bE_w\|^2_{F}\geq\alpha^\frac{2}{K}\right\} \\
& \leq \alpha^{-\frac{2}{K}}\mathsf{E}\!\left[\|\bD_w^{-1}\bE_w\|^2_{F}\right].
\end{align*}
With  $\mathrm{Pr}\!\left\{\|\bD_w^{-1}\bE_w\|^K_{F}<\alpha\right\} = 1-\mathrm{Pr}\!\left\{\|\bD_w^{-1}\bE_w\|^K_{F}\geq\alpha\right\}$ and by using the upper bound for $\mathsf{E}\!\left[\|\bD_w^{-1}\bE_w\|^2_{F}\right]$ from Lemma~\ref{lem3}, we finally obtain  \eqref{eq:thmcondition}.


\bibliographystyle{IEEEtran}
\bibliography{IEEEabrv,Massive_MIMO_bibfile}

\begin{thebibliography}{10}
\providecommand{\url}[1]{#1}
\csname url@samestyle\endcsname
\providecommand{\newblock}{\relax}
\providecommand{\bibinfo}[2]{#2}
\providecommand{\BIBentrySTDinterwordspacing}{\spaceskip=0pt\relax}
\providecommand{\BIBentryALTinterwordstretchfactor}{4}
\providecommand{\BIBentryALTinterwordspacing}{\spaceskip=\fontdimen2\font plus
\BIBentryALTinterwordstretchfactor\fontdimen3\font minus
  \fontdimen4\font\relax}
\providecommand{\BIBforeignlanguage}[2]{{%
\expandafter\ifx\csname l@#1\endcsname\relax
\typeout{** WARNING: IEEEtran.bst: No hyphenation pattern has been}%
\typeout{** loaded for the language `#1'. Using the pattern for}%
\typeout{** the default language instead.}%
\else
\language=\csname l@#1\endcsname
\fi
#2}}
\providecommand{\BIBdecl}{\relax}
\BIBdecl

\bibitem{Wu2012}
M.~Wu, B.~Yin, A.~Vosoughi, C.~Studer, J.~R. Cavallaro, and C.~Dick,
  ``Approximate matrix inversion for high-throughput data detection in the
  large-scale {MIMO} uplink,'' in \emph{Proc. IEEE ISCAS}, Beijing, China, May
  2013, pp. 2155--2158.

\bibitem{Yin2013}
B.~Yin, M.~Wu, C.~Studer, J.~R. Cavallaro, and C.~Dick, ``Implementation
  trade-offs for linear detection in large-scale {MIMO} systems,'' in
  \emph{Proc. IEEE ICASSP}, Vancouver, BC, May 2013, pp. 2679--2683.

\bibitem{Paulraj2008}
A.~Paulraj, R.~Nabar, and D.~Gore, \emph{Introduction to Space-Time Wireless
  Communications}.\hskip 1em plus 0.5em minus 0.4em\relax New York, USA:
  Cambridge University Press, 2008.

\bibitem{3GPPLTE}
\emph{3rd Generation Partnership Project; Technical Specification Group Radio
  Access Network; Evolved Universal Terrestrial Radio Access ({E-UTRA});
  Multiplexing and channel coding (Release 9)}.\hskip 1em plus 0.5em minus
  0.4em\relax 3GPP Organizational Partners TS 36.212 Rev. 8.3.0, May 2008.

\bibitem{SesiaLTE}
S.~Sesia, I.~Toufik, and M.~Baker, \emph{{LTE}, The {UMTS} Long Term Evolution:
  From Theory to Practice}.\hskip 1em plus 0.5em minus 0.4em\relax Wiley
  Publishing, 2009.

\bibitem{3GPPLTEA}
\emph{3rd Generation Partnership Project; Technical Specification Group Radio
  Access Network; Evolved Universal Terrestrial Radio Access ({E-UTRA});
  Physical Layer Procedures (Release 10)}.\hskip 1em plus 0.5em minus
  0.4em\relax 3GPP Organizational Partners TS 36.213 version 10.10.0, Jul.
  2013.

\bibitem{IEEE802.11n}
\emph{{IEEE} Draft Standard Part 11: Wireless {LAN} Medium Access Control
  ({MAC}) and Physical Layer ({PHY}) specifications: Amendment 4: Enhancements
  for Higher Throughput}.\hskip 1em plus 0.5em minus 0.4em\relax
  P802.11n\_D3.00, Sept. 2007.

\bibitem{Marzetta2010}
T.~L. Marzetta, ``Noncooperative cellular wireless with unlimited numbers of
  base station antennas,'' \emph{IEEE Trans. Wireless Commun.}, vol.~9, no.~11,
  pp. 3590--3600, Nov. 2010.

\bibitem{Rusek2012}
F.~Rusek, D.~Persson, B.~K. Lau, E.~G. Larsson, T.~L. Marzetta, O.~Edfors, and
  F.~Tufvesson, ``Scaling up {MIMO}: Opportunities and challenges with very
  large arrays,'' \emph{IEEE Signal Process. Mag.}, vol.~30, no.~1, pp. 40--60,
  Jan. 2013.

\bibitem{Nam2013}
Y.-H. Nam, B.~L. Ng, K.~Sayana, Y.~Li, J.~Zhang, Y.~Kim, and J.~Lee,
  ``Full-dimension {MIMO (FD-MIMO)} for next generation cellular technology,''
  \emph{IEEE Commun. Mag.}, vol.~51, no.~6, pp. 172--179, Jun. 2013.

\bibitem{Huh2011}
H.~Huh, G.~Caire, H.~C. Papadopoulos, and S.~A. Ramprashad, ``Achieving
  ``massive {MIMO}'' spectral efficiency with a not-so-large number of
  antennas,'' \emph{IEEE Trans. Wireless Commun.}, vol.~11, no.~9, pp.
  3266--3239, Sept. 2012.

\bibitem{Ngo2012}
H.~Q. Ngo, E.~G. Larsson, and T.~L. Marzetta, ``Energy and spectral efficiency
  of very large multiuser {MIMO} systems,'' \emph{arXiv preprint: 1112.3810v2},
  May 2012.

\bibitem{agrell2002closest}
E.~Agrell, T.~Eriksson, A.~Vardy, and K.~Zeger, ``Closest point search in
  lattices,'' \emph{IEEE Trans. Inf. Theory}, vol.~48, no.~8, pp. 2201--2214,
  2002.

\bibitem{ABurgThesis}
A.~Burg, ``{VLSI} circuits for {MIMO} communication systems,'' Ph.D.
  dissertation, ETH Z\"urich, Switzerland, 2006.

\bibitem{burg2005vlsi}
A.~Burg, M.~Borgmann, M.~Wenk, M.~Zellweger, W.~Fichtner, and H.~Bolcskei,
  ``{VLSI} implementation of {MIMO} detection using the sphere decoding
  algorithm,'' \emph{IEEE J. Solid-State Circuits}, vol.~40, no.~7, pp.
  1566--1577, Jul. 2005.

\bibitem{wong2002vlsi}
K.~Wong, C.~Tsui, R.~Cheng, and W.~Mow, ``A {VLSI} architecture of a {K-best}
  lattice decoding algorithm for {MIMO} channels,'' in \emph{IEEE ISCAS},
  vol.~3, Scottsdale, AZ, May 2002, pp. 273--276.

\bibitem{HB03}
B.~M. Hochwald and S.~ten Brink, ``Achieving near-capacity on a
  multiple-antenna channel,'' \emph{IEEE Trans. Commun.}, vol.~51, no.~3, pp.
  389--399, 2003.

\bibitem{jsac07}
C.~Studer, A.~Burg, and H.~B\"olcskei, ``Soft-output sphere decoding:
  Algorithms and {VLSI} implementation,'' \emph{IEEE J. Sel. Areas Commun.},
  vol.~26, no.~2, pp. 290--300, Feb. 2008.

\bibitem{jalden2005complexity}
J.~Jald\`en and B.~Ottersten, ``On the complexity of sphere decoding in digital
  communications,'' \emph{IEEE Trans. Signal Process.}, vol.~53, no.~4, pp.
  1474--1484, Apr. 2005.

\bibitem{seethaler2011complexity}
D.~Seethaler, J.~Jald{\'e}n, C.~Studer, and H.~Bolcskei, ``On the complexity
  distribution of sphere decoding,'' \emph{IEEE Trans. Inf. Theory}, vol.~57,
  no.~9, pp. 5754--5768, Sept. 2011.

\bibitem{Datta2012}
T.~Datta, N.~Ashok~Kumar, A.~Chockalingam, and B.~Sundar~Rajan, ``A novel
  {MCMC} algorithm for near-optimal detection in large-scale uplink mulituser
  {MIMO} systems,'' in \emph{Proc. IEEE ITA}, San Diego, CA, Feb. 2012, pp. 69
  --77.

\bibitem{OFDM2004}
R.~Prasad, \emph{{OFDM} for Wireless Communications Systems}.\hskip 1em plus
  0.5em minus 0.4em\relax Norwood, MA, USA: Artech House, Inc., 2004.

\bibitem{3GPP_TS_36.211_v8.6.0}
``{3GPP TS 36.211} {Evolved} {Universal} {Terrestrial} {Radio} {Access}
  ({E-UTRA}) physical channels and modulation (release 8),'' 3rd Generation
  Partnership Project.

\bibitem{khan2009lte}
F.~Khan, \emph{{LTE} for {4G} mobile broadband: air interface technologies and
  performance}.\hskip 1em plus 0.5em minus 0.4em\relax Cambridge University
  Press, 2009.

\bibitem{Studer2011}
C.~Studer, S.~Fateh, and D.~Seethaler, ``{ASIC} implementation of soft-input
  soft-output {MIMO} detection using {MMSE} parallel interference
  cancellation,'' \emph{IEEE J. Solid-State Circuits}, vol.~46, no.~7, pp.
  1754--1765, Jul. 2011.

\bibitem{GV96}
G.~H. Golub and C.~F. {van Loan}, \emph{Matrix Computations}, 3rd~ed.\hskip 1em
  plus 0.5em minus 0.4em\relax The Johns Hopkins Univ. Press, 1996.

\bibitem{fossorier1998equivalence}
M.~P. Fossorier, F.~Burkert, S.~Lin, and J.~Hagenauer, ``On the equivalence
  between {SOVA} and {max-log-MAP} decodings,'' \emph{IEEE Commun. Lett.},
  vol.~2, no.~5, pp. 137--139, May 1998.

\bibitem{burg2006algorithm}
A.~Burg, S.~Haene, D.~Perels, P.~Luethi, N.~Felber, and W.~Fichtner,
  ``Algorithm and {VLSI} architecture for linear {MMSE} detection in
  {MIMO-OFDM} systems,'' in \emph{Proc. IEEE ISCAS}, Island of Kos, Greece, May
  2006, pp. 4102--4105.

\bibitem{rao2010low}
R.~M. Rao, H.~Tarn, R.~Mazahreh, and C.~Dick, ``A low complexity square root
  {MMSE MIMO} decoder,'' in \emph{Proc. 44th Asilomar Conf. on Signals, Systems
  and Computers}, Pacific Grove, CA, Nov. 2010, pp. 1463--1467.

\bibitem{luethi2008gram}
P.~Luethi, C.~Studer, S.~Duetsch, E.~Zgraggen, H.~Kaeslin, N.~Felber, and
  W.~Fichtner, ``{Gram-Schmidt-based} {QR} decomposition for {MIMO} detection:
  {VLSI} implementation and comparison,'' in \emph{Proc. IEEE APCCAS}, Macao,
  China, Nov. 2008, pp. 830--833.

\bibitem{Stewart1998}
G.~Stewart, \emph{Matrix Algorithms: Basic decompositions}, 1998.

\bibitem{hoydis2011massive}
J.~Hoydis, S.~Ten~Brink, and M.~Debbah, ``Massive {MIMO}: How many antennas do
  we need?'' in \emph{49th Ann. Allerton Conf. on Commun., Control., and
  Comput.}, Monticello, IL, Sept. 2011, pp. 545--550.

\bibitem{SamsungFDMIMO}
\emph{New SID Proposal: Study on Full Dimension MIMO for LTE}.\hskip 1em plus
  0.5em minus 0.4em\relax 3GPP TSG RAN Meeting 58, Dec. 2012.

\bibitem{ArgosV2}
C.~Shepard, H.~Yu, and L.~Zhong, ``{ArgosV2}: a flexible many-antenna research
  platform,'' in \emph{Proc.~19th Annual International Conference on Mobile
  Computing and Networking (MobiCom)}.\hskip 1em plus 0.5em minus 0.4em\relax
  Miami, Florida, USA: ACM, 2013, pp. 163--166.

\bibitem{winner2}
\BIBentryALTinterwordspacing
L.~Hentil{\"a}, P.~Ky{\"o}sti, M.~K{\"a}ske, M.~Narandzic, and M.~Alatossava.
  (2007, Dec.) Matlab implementation of the {WINNER} phase {II} channel model
  ver~1.1. [Online]. Available:
  \url{https://www.ist-winner.org/phase\_2\_model.html}
\BIBentrySTDinterwordspacing

\bibitem{HoydisChannel2012}
J.~Hoydis, C.~Hoek, T.~Wild, and S.~ten Brink, ``Channel measurements for large
  antenna arrays,'' in \emph{Proc. IEEE ISWCS}, Aug. 2012.

\bibitem{Simko2011}
M.~Simko, D.~Wu, C.~Mehlfuehrer, J.~Eilert, and D.~Liu, ``Implementation
  aspects of channel estimation for {3GPP LTE} terminals,'' in \emph{Proc. 11th
  European Wireless Conference - Sustainable Wireless Technologies (European
  Wireless)}, Vienna, Austria, Apr. 2011, pp. 440--444.

\bibitem{DFTcorexilinx}
{Xilinx Inc.}, ``{IP LogiCORE Discrete Fourier Transform},'' {DFT IP for 3GPP
  LTE systems}, DS615 v3.1, Mar. 2011.

\bibitem{schreiber1986systolic}
R.~Schreiber and W.-P. Tang, ``On systolic arrays for updating the {Cholesky}
  factorization,'' \emph{BIT Numerical Mathematics}, vol.~26, no.~4, pp.
  451--466, Dec. 1986.

\bibitem{myllyla2005complexity}
M.~Myllyla, J.~Hintikka, J.~R. Cavallaro, M.~Juntti, M.~Limingoja, and
  A.~Byman, ``Complexity analysis of {MMSE} detector architectures for {MIMO
  OFDM} systems,'' in \emph{Proc. 39th Asilomar Conf. on Signals, Systems and
  Computers}, Nov. 2005, pp. 75--81.

\bibitem{eberli2008divide}
S.~Eberli, D.~Cescato, and W.~Fichtner, ``Divide-and-conquer matrix inversion
  for linear {MMSE} detection in {SDR} {MIMO} receivers,'' in \emph{Proc. 26th
  Norchip Conference}, Nov. 2008, pp. 162--167.

\bibitem{WuDi2011}
D.~Wu, J.~Eilert, and D.~Liu, ``Implementation of a high-speed {MIMO}
  soft-output symbol detector for software defined radio,'' \emph{J. Signal
  Process. Syst.}, vol.~63, no.~1, pp. 40--60, Apr. 2011.

\bibitem{karkooti2005fpga}
M.~Karkooti, J.~R. Cavallaro, and C.~Dick, ``{FPGA} implementation of matrix
  inversion using {QRD-RLS} algorithm,'' in \emph{Proc. 44th Asilomar Conf. on
  Signals, Systems and Computers}, Nov. 2005, pp. 1625--1629.

\bibitem{Cong2011}
J.~Cong, B.~Liu, S.~Neuendorffer, J.~Noguera, K.~Vissers, and Z.~Zhang,
  ``High-level synthesis for {FPGAs}: From prototyping to deployment,''
  \emph{IEEE Trans. Comput.-Aided Design Integr. Circuits Syst.}, vol.~30,
  no.~4, pp. 473--491, Apr. 2011.

\bibitem{kim2007efficient}
H.~S. Kim, W.~Zhu, J.~Bhatia, K.~Mohammed, A.~Shah, and B.~Daneshrad, ``An
  efficient {FPGA} based {MIMO-MMSE} detector,'' in \emph{Proc. EUSIPCO}, Sept.
  2007, pp. 1131--1135.

\bibitem{Eilert2008}
J.~Eilert, D.~Wu, and D.~Liu, ``Implementation of a programmable linear {MMSE}
  detector for {MIMO-OFDM},'' in \emph{Proc. IEEE ICASSP}, Mar. 2008, pp.
  5396--5399.

\bibitem{wu2010vlsi}
D.~Wu, J.~Eilert, R.~Asghar, and D.~Liu, ``{VLSI} implementation of a
  fixed-complexity soft-output {MIMO} detector for high-speed wireless,''
  \emph{EURASIP Journal on Wireless Communications and Networking}, vol. 2010,
  pp. 58:1--58:13, Apr. 2010.

\bibitem{Yin2014}
B.~Yin, M.~Wu, G.~Wang, C.~Dick, J.~R. Cavallaro, and C.~Studer, ``A
  3.8\,{Gb/s} large-scale {MIMO} detector for {3GPP LTE-Advanced},'' in
  \emph{Proc. IEEE ICASSP}, Florcence, Italy, May 2014.

\bibitem{SB10}
C.~Studer and H.~B\"olcskei, ``Soft--input soft--output single tree-search
  sphere decoding,'' \emph{IEEE Trans. Inf. Theory}, vol.~56, no.~10, pp.
  4827--4842, Oct. 2010.

\bibitem{cook2008notes}
J.~D. Cook, ``Inverse {Gamma} distribution,'' {online:
  \url{http://www.johndcook.com/inverse\_gamma.pdf}}, Tech. Rep., 2008.

\end{thebibliography}


\end{document}